%% file: syntax.tex
\documentclass[11pt,fleqn]{article}

\usepackage{latexsym}
\usepackage{amssymb}
\usepackage{stmaryrd}
\usepackage{phonetic}
\usepackage{pgf}
\usepackage{tikz}
\usepackage{url}
\usetikzlibrary{arrows}

\input{def}

\title{{\bf Frameworks for Reasoning about Syntax that Utilize Quotation
  and Evaluation}\thanks{This research was supported by NSERC.}}

\author{William M. Farmer and Pouya Larjani\thanks{Address:
  Department of Computing and Software, McMaster University 1280 Main
  Street West, Hamilton, Ontario L8S 4K1, Canada. E-mail:
  \texttt{wmfarmer@mcmaster.ca}, \texttt{pouya.larjani@gmail.com}.}}

\date{24 June 2014}

\begin{document}

\maketitle

\begin{abstract}
It is often useful, if not necessary, to reason about the syntactic
structure of an expression in an interpreted language (i.e., a
language with a semantics).  This paper introduces a mathematical
structure called a \emph{syntax framework} that is intended to be an
abstract model of a system for reasoning about the syntax of an
interpreted language.  Like many concrete systems for reasoning about
syntax, a syntax framework contains a mapping of expressions in the
interpreted language to syntactic values that represent the syntactic
structures of the expressions; a language for reasoning about the
syntactic values; a mechanism called \emph{quotation} to refer to the
syntactic value of an expression; and a mechanism called
\emph{evaluation} to refer to the value of the expression represented
by a syntactic value.  A syntax framework provides a basis for
integrating reasoning about the syntax of the expressions with
reasoning about what the expressions mean.  The notion of a syntax
framework is used to discuss how quotation and evaluation can be built
into a language and to define what quasiquotation is.  Several
examples of syntax frameworks are presented.
\end{abstract}

\iffalse
\fi

\newpage

\section{Introduction} \label{sec:intro}

Every calculus student knows that computing the derivative of a
function directly from the definition is an excruciating task, while
computing the derivative using the rules of differentiation is a
pleasure.  A differentiation rule is a function, but not a usual
function like the square root function or the limit of a sequence
operator.  Instead of mapping a function to its derivative, it maps
one syntactic representation of a function to another.  For example,
the \emph{product rule} maps an expression of the
form \[\frac{d}{dx}(u \cdot v),\] where $u$ and $v$ are expressions
that may include occurrences of $x$, to the
expression \[\frac{d}{dx}(u) \cdot v + u \cdot \frac{d}{dx}(v).\]

We call a mapping, like a differentiation rule, that takes one
syntactic expression to another syntactic expression a
\emph{transformer}~\cite{FarmerMohrenschildt03}.  A full formalization
of calculus requires a reasoning system in which (1) the derivative of
a function can be defined, (2) the differentiation rules can be
represented as transformers, and (3) the transformers representing the
differentiation rules can be shown to compute derivatives.  Such a
reasoning system must provide the means to reason about the syntactic
manipulation of expressions as well as the connection these
manipulations have to the semantics of the expressions.  In other
words, the reasoning system must allow one to reason about syntax and
its relationship to semantics.  See \cite{Farmer13} for a detailed
discussion about the formalization of symbolic differentiation and
other syntax-based mathematical algorithms.

An \emph{interpreted language} is a language $L$ such that each
expression $e$ in $L$ is mapped to a \emph{semantic value} that serves
as the meaning of $e$.  What facilities does a reasoning system need
for reasoning about the interplay of the syntax and semantics of an
interpreted language $L$?  Here are four candidates:

\be

  \item A set of \emph{syntactic values} that represent the syntactic
    structures of the expressions in $L$.

  \item \bsp A language for expressing statements about syntactic
    values and thereby indirectly about the syntactic structures of
    the expressions in $L$.\esp

  \item A mechanism called \emph{quotation} for referring to the
    syntactic value that represents a given expression in $L$.

  \item A mechanism called \emph{evaluation} for referring to the
    semantic value of the expression whose syntactic structure is
    represented by a given syntactic value.

\ee 
Quotation and evaluation together provide the means to integrate
reasoning about the syntax of the expressions with reasoning about
what the expressions mean.

This paper has three objectives.  The first objective is to introduce a
mathematical structure called a \emph{syntax framework} that is
intended to be an abstract model of a system for reasoning about the
syntax of an interpreted language.  A syntax framework for an
interpreted language $L$ contains four components corresponding to the
four facilities mentioned just above:

\be

  \item A function called a \emph{syntax representation} that maps
    each expression $e$ in $L$ to a \emph{syntactic value} that
    represents the syntactic structure of $e$.

  \item A language called a \emph{syntax language} whose expressions
    denote syntactic values.  

  \item A \emph{quotation} function that maps an expression $e$ in $L$
    to an expression in the syntax language that denotes the syntactic
    value of $e$.

  \item An \emph{evaluation} function that maps an expression $e$ in
    the syntax language to an expression in $L$ whose semantic value
    is the same as that of the expression in $L$ whose syntactic value
    is denoted by $e$.

\ee

The second objective is to demonstrate that a syntax framework has the
ingredients needed for reasoning effectively about syntax.  We discuss
the benefits of a syntax framework for reasoning about syntax and
particularly for reasoning about transformers like the differentiation
rules.  We explain how the liar paradox can be avoided when quotation
and evaluation are built-in operators.  And we define in a syntax
framework a notion of quasiquotation which greatly facilitates
constructing expressions that denote syntactic values.

The third objective is to show that the notion of a syntax framework
embodies a common structure that is found in a variety of systems for
reasoning about the interplay of syntax and semantics.  In particular,
we show that the standard systems in which syntactic structure is
represented by strings, G\"odel numbers, and members of an inductive
type are instances of a syntax framework.  We also show that several
more sophisticated systems from the literature, including a simplified
version of Lisp, can be viewed as syntax frameworks.

\emph{Reflection} is a technique to embed reasoning about a reasoning
system (i.e., metareasoning) in the reasoning system itself.
Reflection has been employed in logic~\cite{Koellner09}, theorem
proving~\cite{Harrison95}, and programming~\cite{DemersMalenfant95}.
Since metareasoning very often involves the syntactic manipulation of
expressions, a syntax framework is a natural subcomponent of a
reflection mechanism.

\iffalse
The ideas underlying our notion of a syntax framework are not deep,
but they tend to be confusing since they deal with the interplay of
syntax and semantics.  This confusion is well known to programmers who
are trying to sort out the meanings of \texttt{quote} and
\texttt{eval} in Lisp.  Our objective is to explicate what quotation
and evaluation are and how they interact.  We believe that these ideas
are useful in both logic and programming.  Most of the paper is
devoted to examples and definitions, but some applications of syntax
frameworks are briefly discussed.
\fi

The rest of the paper is organized as follow.  The next section,
section~\ref{sec:syn-frame}, defines the notion of a syntax framework
and discusses it benefits.  Section~\ref{sec:examples} presents three
standard syntax reasoning systems that are instances of a syntax
framework.  Section~\ref{sec:built-in} discusses built-in operators
for quotation and evaluation as found in Lisp and other languages and
explains how the liar paradox is avoided in a syntax framework.
Section~\ref{sec:quasiquotation} defines a notion of quasiquotation in
a syntax framework.  Section~\ref{sec:literature} identifies some
sophisticated syntax reasoning systems in the literature that are
instances of a syntax framework.  The paper ends with a conclusion in
section~\ref{sec:conclusion}.

\section{Syntax Frameworks} \label{sec:syn-frame}

In this section we will define a mathematical structure called a
\emph{syntax framework}.  In the subsequent sections we will give
several examples of syntax reasoning systems that can be interpreted
as instances of this structure.

The reader should note that the notion of a syntax framework presented
here is not adequate to interpret syntax reasoning systems, such as
programming languages, that contain context-sensitive expressions
(such as mutable variables).  To interpret these kinds of systems, a
syntax framework must be extended to a \emph{contextual syntax
  framework} that includes mutable contexts.  For further discussion,
see Remark~\ref{rem:contextual}.

\subsection{Interpreted Languages} \label{subsec:inter-lang}

Let a \emph{formal language} be a set of expressions each having a
unique mathematically precise syntactic structure.  We will leave
``expression'' and ``mathematically precise syntactic structure''
unspecified.  A formal language $L$ is a \emph{sublanguage} of a
formal language $L'$ if $L \subseteq L'$.

An interpreted language is a formal language with a semantics:

\begin{df}[Interpreted Language] \label{df:interp-lang} \em \bsp
An \emph{interpreted language} is a triple $I=(L,D_{\rm sem},V_{\rm
  sem})$ where: 

\be

  \item $L$ is a formal language.

  \item $D_{\rm sem}$ is a nonempty domain (set) of \emph{semantic
    values}.

  \item $V_{\rm sem} : L \tarrow D_{\rm sem}$ is a total function,
    called a \emph{semantic valuation function}, that assigns each
    expression $e \in L$ a semantic value $V_{\rm sem}(e) \in D_{\rm
      sem}$. \hfill $\Box$

\ee 
\esp
\end{df}
An interpreted language is thus a formal language with an associated
assignment of a semantic meaning to each expression in the language.
Each expression of an interpreted language thus has both a syntactic
structure and a semantic meaning.  There is no restriction placed on
what can be a semantic value.  An interpreted language is graphically
depicted in Figure~\ref{fig:interp-lang} (we will add elements to this
figure as the discussion advances).

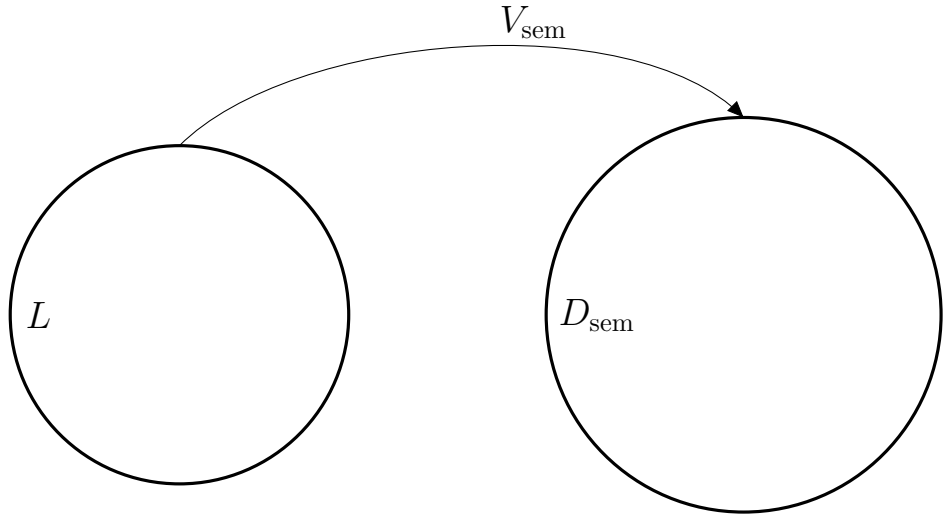
\begin{figure}
\center
\begin{tikzpicture}[scale=.75]
  %\draw[step=.5] (-5,-5) grid (8,5);
  \draw[very thick] (-3,0) circle (3);
    \draw (-5.5,0) node {\Large $L$};
  \draw[very thick] (7,0) circle (3.5);
    \draw (4.4,0) node {\Large $D_{\rm sem}$};
  \draw[-triangle 45] (-3,3) .. controls (-1,5) and (5,5.5) .. (7,3.5);
    \draw[right] (2.5,5.2) node {\Large $V_{\rm sem}$};
\end{tikzpicture}
\caption{An Interpreted Language}  \label{fig:interp-lang}
\end{figure}

\begin{eg}[Many-Sorted First-Order Languages] \label{eg:ms-fol} \em
\bsp Let $L$ be the set of the terms and formulas of a many-sorted
first-order language with sorts $\alpha_1,\ldots,\alpha_n$.  Define
$L_i$ to be the set of terms of sort $\alpha_i$ for each $i$ with $1
\le i \le n$ and $L_{\rm f}$ to be the set of formulas of the
many-sorted first-order language. \esp

Let $(D_1,\ldots,D_n,I)$ be a model for the many-sorted first-order
language $L$ where each $D_i$ is a nonempty domain and $I$ is an
interpretation function for the individual constants, function
symbols, and predicate symbols of $L$.  Let $\phi_i$ be a mapping from
the variables in $L_i$ to $D_i$ for each $i$ with $1 \le i \le n$.
The model $(D_1,\ldots,D_n,I)$ and variable assignments
$\phi_1,\ldots,\phi_n$ determine a semantic valuation function $V_i :
L_i \tarrow D_i$ on terms of sort $\alpha_i$ for each $i$ with $1 \le
i \le n$ and a semantic valuation function $V_{\rm f} : L_{\rm f}
\tarrow \set{\TRUE,\FALSE}$ on formulas.  Then \[(L,D_1 \cup \cdots
\cup D_n \cup \set{\TRUE,\FALSE}, V_1 \cup \cdots \cup V_n \cup V_{\rm
  f})\] is an interpreted language. \hfill $\Box$
\end{eg}

\subsection{Syntax Representations and Syntax Languages}

A syntax representation of a formal language is an assignment of
syntactic values to the expressions of the language:

\begin{df}[Syntax Representation] \label{df:syn-rep} \em \bsp
Let $L$ be a formal language. A \emph{syntax representation} of $L$ is
a pair $R=(D_{\rm syn},V_{\rm syn})$ where:

\be

  \item $D_{\rm syn}$ is a nonempty domain (set) of \emph{syntactic
    values}.  Each member of $D_{\rm syn}$ represents a syntactic
    structure.

  \item $V_{\rm syn} : L \tarrow D_{\rm syn}$ is an injective, total
    function, called a \emph{syntactic valuation function}, that
    assigns each expression $e \in L$ a syntactic value $V_{\rm
      syn}(e) \in D_{\rm syn}$ such that $V_{\rm syn}(e)$ represents
    the syntactic structure of $e$. \hfill $\Box$

\ee 
\esp
\end{df}
A syntax representation of a formal language is thus an assignment of
a syntactic meaning to each expression in the language.  Notice that,
if $R=(D_{\rm syn},V_{\rm syn})$ is a syntax representation of $L$,
then $(L,D_{\rm syn},V_{\rm syn})$ is an interpreted language.

\begin{eg}[Expressions as Strings: Syntax Representation] \label{eg:strings-a} \em \bsp
Let $L$ be a many-sorted first-order language.  The expressions of $L$
--- i.e.,~the terms and formulas of $L$ --- can be viewed as certain
strings of symbols.  For example, the term $f(x)$ can be viewed as the
string \texttt{"f(x)"} composed of four symbols.  Let $\sA$ be the
alphabet of symbols occurring in the expressions of $L$ and
$\mname{strings}_{\cal A}$ be the set of strings over $\sA$.  Then the
syntactic structure of an expression can be represented by a string in
$\mname{strings}_{\cal A}$, and we can define a function $S : L
\tarrow \mname{strings}_{\cal A}$ that maps each expression of $L$ to
the string over $\sA$ that represents its syntactic structure.  $S$ is
an injective, total function since, for each $e \in L$, there is
exactly one string in $\mname{strings}_{\cal A}$ that represents the
syntactic structure of $e$.  Therefore, $(\mname{strings}_{\cal A},
S)$ is a syntax representation of $L$. \hfill $\Box$ \esp
\end{eg}

A syntax language for a syntax representation is a language of
expressions that denote syntactic values in the syntax representation:

\begin{df}[Syntax Language] \label{df:syn-lang} \em \bsp
Let $R=(D_{\rm syn},V_{\rm syn})$ be a syntax representation of a
formal language $L_{\rm obj}$.  A \emph{syntax language} for $R$ is a pair
$(L_{\rm syn}, I)$ where:

\be

  \item $I = (L,D_{\rm sem},V_{\rm sem})$ in an interpreted language.

  \item $L_{\rm obj}\subseteq L$, $L_{\rm syn} \subseteq L$, and
    $D_{\rm syn} \subseteq D_{\rm sem}$.

  \item $V_{\rm sem}$ restricted to $L_{\rm syn}$ is a total function
    $V'_{\rm sem} : L_{\rm syn} \tarrow D_{\rm syn}$. \hfill $\Box$

\ee
\esp
\end{df}
Notice that, if $(L_{\rm syn}, I)$ is a syntax language for $R$ (as
in the definition above), then $(L_{\rm syn}, D_{\rm syn}, V'_{\rm sem})$ is an
interpreted language.

\begin{eg}[Expressions as Strings: Syntax Language] \label{eg:strings-b} \em
\bsp Let $I= (L, D, V)$ where \[D = D_1 \cup \cdots \cup D_n \cup
\set{\TRUE,\FALSE}\] and \[V = V_1 \cup \cdots \cup V_n \cup V_{\rm
  f}\] be the interpreted language given in Example~\ref{eg:ms-fol}.
Recall that $L$ is the set of terms and formulas of a many-sorted
first-order language with sorts $\alpha_1,\ldots,\alpha_n$.  Suppose
$\alpha_1 = \mname{Symbol}$, $\alpha_2 = \mname{String}$, $D_1$ is the
alphabet of $L$, and $D_2$ is the set of strings over $D_1$.  Let $S :
L \tarrow D_2$ be the total function that maps each $e \in L$ to the
string in $D_2$ that represents the syntactic structure of $e$.  Then
$R= (D_2,S)$ is a syntax representation of $L$ as in
Example~\ref{eg:strings-a} and $(L_2,I)$ is a syntax language for $R$
since $L_2 \subseteq L$, $D_2 \subseteq D$, and $V$ restricted to
$L_2$ is $V_2 : L_2 \tarrow D_2$. \hfill $\Box$ \esp
\end{eg}

\subsection{Definition of a Syntax Framework} \label{subsec:frameworks}

A syntax framework is a structure that is built from an interpreted
language $I = (L,D_{\rm sem},V_{\rm sem})$ in three stages.  

The first stage is to choose an object language $L_{\rm obj} \subseteq
L$ and a syntax representation $R=(D_{\rm syn},V_{\rm syn})$ for
$L_{\rm obj}$ such that $D_{\rm syn} \subseteq D_{\rm sem}$.  ($L_{\rm
  obj}$ could be the entire language $L$ as in
Example~\ref{eg:strings-a}.)  This first stage is depicted in
Figure~\ref{fig:syn-frame-stage-1}.

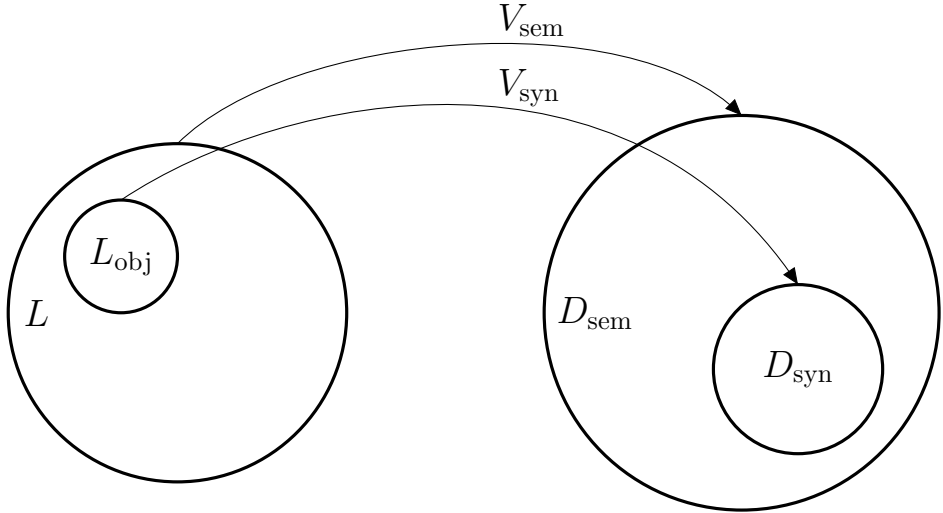
\begin{figure}
\center
\begin{tikzpicture}[scale=.75]
  %\draw[step=.5] (-5,-5) grid (8,5);
  \draw[very thick] (-3,0) circle (3);
    \draw (-5.5,0) node {\Large $L$};
  \draw[very thick] (-4,1) circle (1);
    \draw (-4,1) node {\Large $L_{\rm obj}$};
  \draw[very thick] (7,0) circle (3.5);
    \draw (4.4,0) node {\Large $D_{\rm sem}$};
  \draw[very thick] (8,-1) circle (1.5);
    \draw (8,-1) node {\Large $D_{\rm syn}$};
  \draw[-triangle 45] (-3,3) .. controls (-1,5) and (5,5.5) .. (7,3.5);
    \draw[right] (2.5,5.2) node {\Large $V_{\rm sem}$};
  \draw[-triangle 45] (-4,2) .. controls (-1,4) and (5,5) .. (8,.5);
    \draw[right] (2.5,4) node {\Large $V_{\rm syn}$};
\end{tikzpicture}
\caption{Stage 1 of a Syntax Framework}  \label{fig:syn-frame-stage-1}
\end{figure}

The second stage is to choose a language $L_{\rm syn} \subseteq L$
such that $(L_{\rm syn},I)$ is a syntax language for $R$.  This second
stage, depicted in Figure~\ref{fig:syn-frame-stage-2}, establishes
$L_{\rm syn}$ as a language that can be used to make statements in $L$
about the syntax of the object language $L_{\rm obj}$ via the syntax
representation established in the stage 1.  ($V'_{\rm sem}$ is $V_{\rm
  sem}$ restricted to $L_{\rm syn}$.)

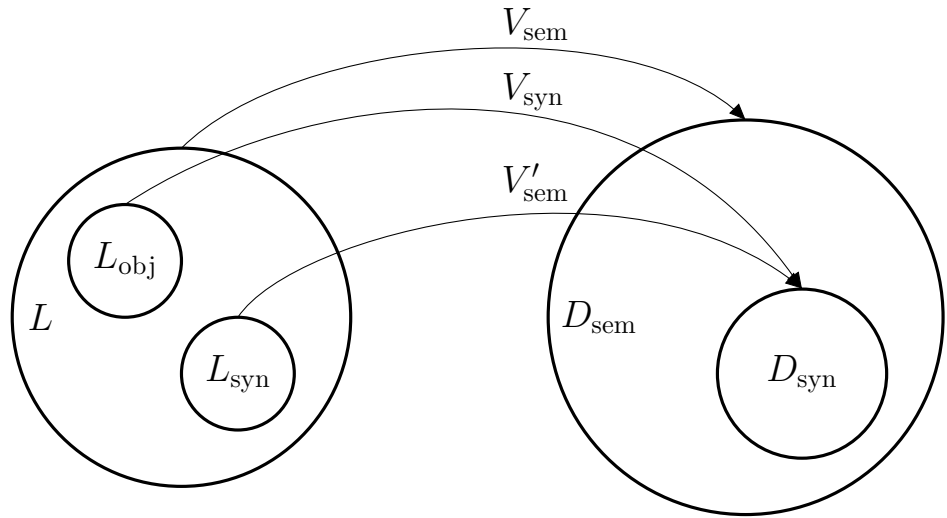
\begin{figure}
\center
\begin{tikzpicture}[scale=.75]
  %\draw[step=.5] (-5,-5) grid (8,5);
  \draw[very thick] (-3,0) circle (3);
    \draw (-5.5,0) node {\Large $L$};
  \draw[very thick] (-4,1) circle (1);
    \draw (-4,1) node {\Large $L_{\rm obj}$};
  \draw[very thick] (-2,-1) circle (1);
    \draw (-2,-1) node {\Large $L_{\rm syn}$};
  \draw[very thick] (7,0) circle (3.5);
    \draw (4.4,0) node {\Large $D_{\rm sem}$};
  \draw[very thick] (8,-1) circle (1.5);
    \draw (8,-1) node {\Large $D_{\rm syn}$};
  \draw[-triangle 45] (-3,3) .. controls (-1,5) and (5,5.5) .. (7,3.5);
    \draw[right] (2.5,5.2) node {\Large $V_{\rm sem}$};
  \draw[-triangle 45] (-4,2) .. controls (-1,4) and (5,5) .. (8,.5);
    \draw[right] (2.5,4) node {\Large $V_{\rm syn}$};
  \draw[-triangle 45] (-2,0) .. controls (-1,1.5) and (5,3) .. (8,.51);
    \draw[right] (2.5,2.4) node {\Large $V'_{\rm sem}$};
\end{tikzpicture}
\caption{Stage 2 of a Syntax Framework}  \label{fig:syn-frame-stage-2}
\end{figure}

The third and final stage is to link $L_{\rm obj}$ and $L_{\rm syn}$
using mappings $Q : L_{\rm obj} \tarrow L_{\rm syn}$ and $E : L_{\rm
  syn} \tarrow L_{\rm obj}$ as depicted in Figure~\ref{fig:syn-frame}.
$Q$ is an injective, total function such that, for all $e \in L_{\rm
  obj}$, \[V_{\rm sem}(Q(e)) = V_{\rm syn}(e).\] For $e \in L_{\rm
  obj}$, $Q(e)$ is called the \emph{quotation} of $e$.  $Q(e)$ denotes
a value in $D_{\rm syn}$ that represents the syntactic structure of
$e$.  $E$ is a (possibly partial) function such that, for all $e \in
L_{\rm syn}$, \[V_{\rm sem}(E(e)) = V_{\rm sem}(V_{\rm
  syn}^{-1}(V_{\rm sem}(e)))\] whenever $E(e)$ is defined.  For $e \in
L_{\rm syn}$, $E(e)$ is called the \emph{evaluation} of $e$.  If it is
defined, $E(e)$ denotes the same value in $D_{\rm sem}$ that the
expression represented by the value of $e$ denotes.  Notice that the
equation above implies $E(e)$ is undefined if $V_{\rm sem}(e)$ is not
in the image of $L_{\rm obj}$ under $V_{\rm syn}$.  Since there will
usually be different $e_1,e_2 \in L_{\rm syn}$ that denote the same
syntactic value, $E$ will usually not be injective.

\begin{figure}
\center
\begin{tikzpicture}[scale=.75]
  %\draw[step=.5] (-5,-5) grid (8,5);
  \draw[very thick] (-3,0) circle (3);
    \draw (-5.5,0) node {\Large $L$};
  \draw[very thick] (-4,1) circle (1);
    \draw (-4,1) node {\Large $L_{\rm obj}$};
  \draw[very thick] (-2,-1) circle (1);
    \draw (-2,-1) node {\Large $L_{\rm syn}$};
  \draw[very thick] (7,0) circle (3.5);
    \draw (4.4,0) node {\Large $D_{\rm sem}$};
  \draw[very thick] (8,-1) circle (1.5);
    \draw (8,-1) node {\Large $D_{\rm syn}$};
  \draw[-triangle 45] (-3,3) .. controls (-1,5) and (5,5.5) .. (7,3.5);
    \draw[right] (2.5,5.2) node {\Large $V_{\rm sem}$};
  \draw[-triangle 45] (-4,2) .. controls (-1,4) and (5,5) .. (8,.5);
    \draw[right] (2.5,4) node {\Large $V_{\rm syn}$};
  \draw[-triangle 45] (-2,0) .. controls (-1,1.5) and (5,3) .. (8,.51);
    \draw[right] (2.5,2.4) node {\Large $V'_{\rm sem}$};
  \draw[-triangle 45] (-3,1) -- (-2,0);
    \draw[right] (-2.6,.8) node {\Large $Q$};
  \draw[-triangle 45] (-3,-1) -- (-4,0);
    \draw[right] (-4.3,-.9) node {\Large $E$};
\end{tikzpicture}
\caption{A Syntax Framework}  \label{fig:syn-frame}
\end{figure}
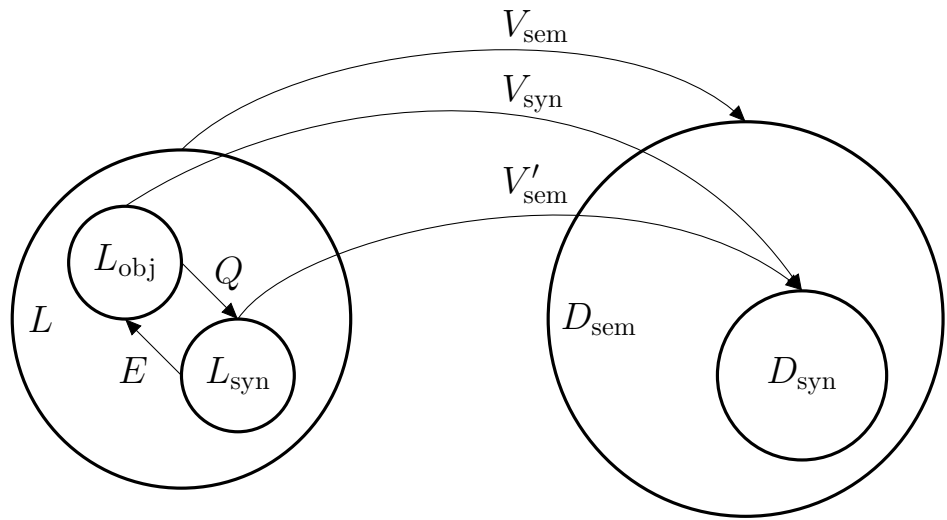

The full definition of a syntax framework is obtained when we put these
three stages together:

\begin{df}[Syntax Framework in an Interpreted Language]\label{df:syn-frame-lang}\em
\bsp
Let $I=(L,D_{\rm sem},V_{\rm sem})$ be an interpreted language
and $L_{\rm obj}$ be a sublanguage of $L$.  A \emph{syntax framework}
for $(L_{\rm obj},I)$ is a tuple $F=(D_{\rm syn},V_{\rm syn}, L_{\rm
  syn}, Q, E)$ where:\esp

\be

  \item $R = (D_{\rm syn},V_{\rm syn})$ is a syntax representation of
    $L_{\rm obj}$.

  \item $(L_{\rm syn},I)$ is syntax language for $R$.

  \item $Q : L_{\rm obj} \tarrow L_{\rm syn}$ is an injective, total
    function, called a \emph{quotation function}, such that:

    \textbf{Quotation Axiom.} For all $e \in L_{\rm obj}$, \[V_{\rm
      sem}(Q(e)) = V_{\rm syn}(e).\]

  \item $E : L_{\rm syn} \tarrow L_{\rm obj}$ is a (possibly partial)
    function, called an \emph{evaluation function}, such that:

    \textbf{Evaluation Axiom.} For all $e \in L_{\rm syn}$, \[V_{\rm
      sem}(E(e)) = V_{\rm sem}(V_{\rm syn}^{-1}(V_{\rm sem}(e)))\]
    whenever $E(e)$ is defined. \hfill $\Box$

\ee 
\end{df}
\bsp \noindent $L$ is called the \emph{full language} of the $F$.
When $D_{\rm sem}$ and $V_{\rm sem}$ are understood, we will say that
$F$ is a syntax framework for $L_{\rm obj}$ in $L$.  Notice that a
syntax framework contains three interpreted languages: $(L,D_{\rm
  sem},V_{\rm sem})$, $(L_{\rm obj},D_{\rm syn},V_{\rm syn})$, and
$(L_{\rm syn}, D_{\rm syn}, V'_{\rm sem})$.  Notice also that the
functions $Q$ and $E$ are part of the metalanguage of $L$ and the
expressions of the form $Q(e)$ and $E(e)$ are not necessarily
expressions of $L$.  In section~\ref{sec:built-in} we will discuss
syntax frameworks in which quotations and evaluations are expressions
in $L$ itself. \esp

\subsection{Two Basic Lemmas}

\bsp Let $I=(L,D_{\rm sem},V_{\rm sem})$ be an interpreted language,
$L_{\rm obj}$ be a sublanguage of $L$, and $F=(D_{\rm syn}, V_{\rm
  syn}, L_{\rm syn}, Q, E)$ be a syntax framework for $(L_{\rm
  obj},I)$. \esp

\begin{lem}[Law of Disquotation] \label{lem:disquotation}
For all $e \in L_{\rm obj}$, \[V_{\rm sem}(E(Q(e))) = V_{\rm sem}(e)\]
whenever $E(Q(e))$ is defined.
\end{lem}

\begin{proof}
Let $e \in L_{\rm obj}$ such that $E(Q(e))$ is defined.  Then
\setcounter{equation}{0}
\begin{eqnarray}
V_{\rm sem}(E(Q(e))) & = & V_{\rm sem}(V_{\rm syn}^{-1}(V_{\rm sem}(Q(e)))) \\
                    & = & V_{\rm sem}(V_{\rm syn}^{-1}(V_{\rm syn}(e))) \\
                    & = & V_{\rm sem}(e)
\end{eqnarray}
(1) follows from the Evaluation Axiom since $E(Q(e))$ is defined. (2)
follows from the Quotation Axiom.  And (3) is by the fact that $V_{\rm
  syn}(e)$ is total on $L_{\rm obj}$.
\end{proof}

\bigskip

The Law of Disquotation does not hold universally in general because
$E$ may not be total on quotations.

\begin{df}[Direct Evaluation] \label{df:direct-eval} \em \bsp
Let $E^{\ast} : L_{\rm syn} \tarrow L_{\rm obj}$ to be the (possibly
partial) function such that, for all $e \in L_{\rm syn}$, $E^{\ast}(e)
= V_{\rm syn}^{-1}(V_{\rm sem}(e))$ whenever $V_{\rm syn}^{-1}(V_{\rm
  sem}(e))$ is defined.  $E^{\ast}$ is called the \emph{direct
  evaluation function for $F$}. \hfill $\Box$ \esp
\end{df}

\begin{lem}[Direct Evaluation] \label{lem:direct-eval}
\be

  \item[]

  \item $E^{\ast}$ satisfies the Evaluation Axiom.

  \item For all $e \in L_{\rm syn}$, if $E^{\ast}(e)$ and $E(e)$ are
    defined, then \[V_{\rm sem}(E^{\ast}(e)) = V_{\rm sem}(E(e)).\]

  \item If $V_{\rm syn}$ is surjective, then $E^{\ast}$ is total.

\ee
\end{lem}

\begin{proof}

\bigskip

\noindent \textbf{Part 1} \sglsp Follows immediate from the definition
of $E^{\ast}$.

\bigskip

\noindent \textbf{Part 2} \sglsp Let $e \in L_{\rm syn}$ such that
$E^{\ast}(e)$ and $E(e)$ are defined.  By the Evaluation Axiom,
$V_{\rm sem}(E(e)) = V_{\rm sem}(V_{\rm syn}^{-1}(V_{\rm sem}(e))) =
V_{\rm sem}(E^{\ast}(e))$.

\bigskip

\noindent \textbf{Part 3} \sglsp Let $V_{\rm syn}$ be surjective and
$e \in L_{\rm syn}$.  Then $V_{\rm syn}^{-1}(V_{\rm sem}(e))$ is
defined and hence $E^{\ast}$ is total by its definition.
\end{proof}

\bigskip

Thus the direct evaluation function is a special evaluation function
that is defined for every syntax framework and is total if the
syntactic valuation function is surjective.

\subsection{Syntax Frameworks in an Interpreted Theory}

The notion of a syntax framework can be easily lifted from an
interpreted language to an interpreted theory.  Let a \emph{theory} be
a pair $T = (L,\Gamma)$ where $L$ is a language and $\Gamma$ is a set
of sentences in $L$ (that serve as the axioms of theory).  A
\emph{model} of $T$ is a pair $M = (D^{M}_{\rm sem},V^{M}_{\rm sem})$
such that $D^{M}_{\rm sem}$ is a set of values that includes the truth
values $\TRUE$ (true) and $\FALSE$ (false) and $V^{M}_{\rm sem} : L
\tarrow D^{M}_{\rm sem}$ is a total function such that, for all
sentences $A \in \Gamma$, $V^{M}_{\rm sem}(A) = \TRUE$.  An
\emph{interpreted theory} is then a pair $I=(T,\sM)$ where $T$ is a
theory and $\sM$ is a set of models of $T$.

A syntax framework in an interpreted theory is a syntax framework with
respect to each model of the interpreted theory:

\begin{df}[Syntax Framework in an Interpreted Theory] \label{df:syn-frame-thy} \em
\hspace{2ex}\\
Let $I=(T,\sM)$ be an interpreted theory where $T = (L,\Gamma)$ and
$L_{\rm obj}$ be a sublanguage of $L$.  A \emph{syntax framework} for
$(L_{\rm obj},I)$ is a triple $F=(L_{\rm syn}, Q, E)$ where:

\be

  \item $L_{\rm syn} \subseteq L$.

  \item $Q : L_{\rm obj} \tarrow L_{\rm syn}$ is an injective, total
    function.

  \item $E : L_{\rm syn} \tarrow L_{\rm obj}$ is a (possibly partial)
    function.

  \item For all $M = (D^{M}_{\rm sem},V^{M}_{\rm sem}) \in \sM$,
    $F^M=(D^{M}_{\rm syn},V^{M}_{\rm syn}, L_{\rm syn}, Q, E)$ is a
    syntax framework for $(L_{\rm obj},(L,D^{M}_{\rm sem},V^{M}_{\rm
      sem}))$ where $D^{M}_{\rm syn}$ is the range of $V^{M}_{\rm sem}$
    restricted to $L_{\rm syn}$ and $V^{M}_{\rm syn} = V^{M}_{\rm sem}
    \circ Q$.\hfill $\Box$

\ee 
\end{df}

\subsection{Benefits of a Syntax Framework}

The purpose of a syntax framework is to provide the means to reason
about the syntax of a designated object language.  We will briefly
examine the specific benefits that a syntax framework offers for this
purpose.

Let $I=(L,D_{\rm sem},V_{\rm sem})$ be an interpreted language,
$L_{\rm obj}$ be a sublanguage of $L$, and $F=(D_{\rm syn}, V_{\rm
  syn}, L_{\rm syn}, Q, E)$ be a syntax framework for $(L_{\rm
  obj},I)$.

The first, and most important, benefit of $F$ is that it provides a
language, $L_{\rm syn}$, for expressing statements in $L$ about the
syntactical structure of expressions in $L_{\rm obj}$.  These
statements refer to the syntax of $L_{\rm obj}$ via the syntax
representation of $F$.  For example, if $A$ is a formula in $L_{\rm
  obj}$, $e_A$ is an expression in $L_{\rm syn}$ that denotes the
representation of $A$, and $L$ is sufficiently expressive, we could
express in $L$ a statement of the form $\mname{is-implication}(e_A)$
that \emph{indirectly} says ``$A$ is an implication''.

\bsp Having quotation in $F$ enables statements about the syntax of
$L_{\rm obj}$ to be expressed directly in the metalanguage of $L$.
For example, $\mname{is-implication}(Q(A))$ would \emph{directly} say
``$A$ is an implication''.  Quotation also allows us to construct new
expressions from deconstructed components of old expressions.  For
example, if $A \ImpliesAlt B$ is a formula in $L_{\rm obj}$ and $L$ is
sufficiently expressive,
\[\mname{build-implication}(\mname{succedent}(Q(A \ImpliesAlt
B)),\mname{antecedent}(Q(A \ImpliesAlt B)))\] would denote the
representation of $B \ImpliesAlt A$. \esp

Having evaluation in $F$ enables statements about the semantics of the
expressions represented by members of $D_{\rm syn}$ to be expressed
directly in the metalanguage of $L$.  For example, if $c$ is the
expression given in the previous paragraph, then $E(c)$ would be a
formula in $L_{\rm obj}$ that asserts $B \ImpliesAlt A$.

By virtue of these basic benefits, a syntax framework is well equipped
to define and specify transformers.  As we have mentioned in the
introduction, a \emph{transformer} maps expressions to expressions.
More precisely, an \emph{$n$-ary transformer over a language $L$} maps
expressions $e_1,\ldots,e_n$ in $L$ to an expression $e$ in $L$ (where
$n \ge 0$).  A transformer can be defined by either an algorithm
(e.g., a program in a programming language) or a function (e.g., an
expression in a logic that denotes a function).  Transformers include
symbolic computation rules (like the product rule mentioned in the
Introduction), rules of inference, rewrite rules, expression
simplifiers, substitution operations, decision procedures, etc.

A transformer over a language $L$ is usually defined only in the
metalanguage of $L$ and is not defined by an expression in $L$ itself.
For example, the rules of inference for first-order logic are not
expressions in first-order logic.  A syntax framework with a
sufficiently expressive language can be used to transfer a transformer
over $L$ from the metalanguage of $L$ to $L$ itself.  To see this, let
$T : L_{\rm obj} \times \cdots \times L_{\rm obj} \tarrow L_{\rm obj}$
be an $n$-ary transformer over $L_{\rm obj}$ defined in the
metalanguage of $L$.  If $L$ is sufficiently expressive, it would be
possible to define an operator $e_T : L_{\rm syn} \times \cdots \times
L_{\rm syn} \tarrow L_{\rm syn}$ in $L$ that denotes a function $f_T :
D_{\rm syn} \times \cdots \times D_{\rm syn} \tarrow D_{\rm syn}$ that
represents $T$.  Using quotation, $e_T$ is specified by the following
statement in the metalanguage of $L$: \[\ForallApp e_1,\ldots,e_n
\mcolon L_{\rm obj} \mdot e_T(Q(e_1),\ldots,Q(e_n)) =
Q(T(e_1,\ldots,e_n)).\]

The full power of a syntax framework is exhibited in a specification
of the semantic meaning of a transformer.  Suppose $L$ is a language
of natural number arithmetic, the expressions in $L_{\rm obj}$ denote
natural numbers, $L_{\rm obj}$ contains a sublanguage $L_{\rm nat}$ of
terms denoting natural numbers, and $L_{\rm syn}$ contains a
sublanguage $L_{\rm num}$ of terms denoting natural number numerals
$Q(0),Q(1),Q(2),\ldots$.  Further suppose that $\mname{add}$ is a
binary transformer over $L_{\rm nat}$ that ``adds'' two natural number
terms so that, e.g., $\mname{add}(2,3) = 5$.  Then, using evaluation,
the semantic meaning of $e_{\sf add}$, the representation of
\mname{add} in $L$, is specified by the following statement in the
metalanguage of $L$:
\[\ForallApp e_1,e_2 \mcolon L_{\rm num} \mdot E(e_{\sf add}(e_1,e_2)) = 
E(e_1) + E(e_2)\] where $+ : L_{\rm nat} \times L_{\rm nat} \tarrow
L_{\rm nat}$ is a binary operator in $L$ that denotes the sum
function.

See \cite{Farmer13} for further discussion on how transformers can be
formalized using a syntax framework.

\subsection{Further Remarks}

\begin{rem}[Syntax Representation]\em
Although a syntax representation is a crucial component of a syntax
framework, very little restriction is placed on what a syntax
representation can be.  Almost any representation that captures the
syntactic structure of the expressions in the object language is
acceptable.  In fact, it is not necessary to capture the entire
syntactic structure of an expression, only the part of the syntactic
structure that is of interest to the developer of the syntax
framework.\hfill $\Box$
\end{rem}

\begin{rem}[Theories of Quotation]\em
The quotation function $Q$ of a syntax framework is based on the
\emph{disquotational theory of quotation}~\cite{Quotation12}.
According to this theory, a quotation of an expression $e$ is an
expression that denotes $e$ itself.  In our definition of a syntax
framework, $Q(e)$ denotes a value that represents $e$ (as a syntactic
entity).  Andrew Polonsky presents in~\cite{Polonsky11} a set of
axioms for quotation operators of this kind.  There are several other
theories of quotation that have been
proposed~\cite{Quotation12}.\hfill $\Box$
\end{rem}

\begin{rem}[Theories of Truth]\em\bsp
When $e$ is a representation of a truth-valued expression $e'$, the
evaluation $E(e)$ is a formula that asserts the truth of $e'$.  Thus
the evaluation function $E$ of a syntax framework is a \emph{truth
  predicate}~\cite{Truth13}.  A truth predicate is the face of a
\emph{theory of truth}: the properties of a truth predicate
characterize a theory of truth~\cite{Leitgeb07}.  The definition of a
syntax framework imposes no restriction on $E$ as a truth predicate
other than that the Evaluation Axiom must hold.  What truth is and how
it can be formalized is a fundamental research area of logic, and
avoiding inconsistencies derived from the liar paradox (which we
address below) and similar statements is one of the major research
issues in the area (see~\cite{Halbach11}).\hfill $\Box$\esp
\end{rem}

\begin{rem}[Contextual Syntax Frameworks]\em\label{rem:contextual}
We have mentioned already that a syntax framework cannot interpret
syntax reasoning systems that contain context-sensitive expressions.
This means that a syntax framework is not suitable for programming
languages with mutable variables.  For programming languages, a syntax
framework needs to be generalized to a \emph{contextual syntax
  framework} that includes a semantic valuation function that takes a
\emph{valuation context} as part of its input and returns a modified
valuation context as part of its output.  \emph{Metaprogramming} is
the writing of programs that manipulate other programs.  It requires a
means to manipulate the syntax of the programs in a programming
language.  In other words, metaprogramming requires code to be data.
Examples of metaprogramming languages include Lisp,
Agda~\cite{Norell07,Norell09}, F\#~\cite{FSharp11},
MetaML~\cite{TahaSheard00}, MetaOCaml~\cite{MetaOCaml11},
reFLect~\cite{GrundyEtAl06}, and Template
Haskell~\cite{SheardJones02}.  An appropriate contextual syntax
framework would provide a good basis for discussing the code
manipulation done in metaprogramming.  We will present the notion of a
contextual syntax framework in a future paper.\hfill $\Box$
\end{rem}

\section{Three Standard Examples} \label{sec:examples}

We will now present three standard syntax reasoning systems that are
examples of a syntax framework.

\subsection{Example: Expressions as Strings} \label{subsec:strings}

We will continue the development of Example~\ref{eg:strings-b}.
Suppose $L$ contains the following operators:

\bi

  \item An individual constant $c_a$ of sort \mname{Symbol} for each
    $a \in \sA$.

  \item An individual constant \mname{nil} of sort \mname{String}.

  \item A function symbol \mname{cons} of sort $\mname{Symbol} \times
    \mname{String} \tarrow \mname{String}$.

  \item A function symbol \mname{head} of sort $\mname{String} \tarrow
    \mname{Symbol}$.

  \item A function symbol \mname{tail} of sort $\mname{String} \tarrow
    \mname{String}$.

\ei
The terms of sort \mname{String} are intended to denote strings over
$\sA$ in the usual way.  \mname{cons} is used to describe the
construction of strings, while \mname{head} and \mname{tail} are used
to describe the deconstruction of strings.  The terms of sort
\mname{String} can thus be used as a language to reason
\emph{directly} about strings over $\sA$ and \emph{indirectly} about
the syntactic structure of the expressions of $L$ (including the terms
of sort \mname{String} themselves).

This reasoning system for the syntax of $L$ can be strengthened by
interconnecting the expressions of $L$ and the terms of sort
\mname{String}.  This is done by defining a quotation function $Q$ and
an evaluation function $E$.  

$Q : L \tarrow L_2$ maps each expression $e$ of $L$ to a term $Q(e)$
of sort \mname{String} such that $Q(e)$ denotes $S(e)$, the string
over $\sA$ that represents $e$.  For example, $Q$ could map $f(x)$
to \[\mname{cons}(c_{\tt f}, \mname{cons}(c_{\tt (},
\mname{cons}(c_{\tt x}, \mname{cons}(c_{\tt )}, \mname{nil})))),\]
which denotes the string \texttt{"f(x)"}.  Thus $Q$ provides the means
to refer to a representation of the syntactic structure of an
expression of $L$.

$E : L_2 \tarrow L$ maps each term $t$ of sort \mname{String} to the
expression $E(t)$ of $L$ such that the syntactic structure of $E(t)$
is represented by the string denoted by $t$ provided $t$ denotes a
string that actually represents the syntactic structure of some
expression of $L$.  For example, $E$ maps the term displayed above
(i.e., $Q(f(x))$) to $f(x)$.  Thus $E$ provides the means to refer to
the value of the expression whose syntactic structure is represented
by the string that a term of sort \mname{String} denotes.  $E$ is a
partial function on the terms of sort \mname{String} since not every
string in $D_2$ represents the syntactic structure of some expression
in $L$ and $V_2$ is surjective.  Notice that, for all expressions $e$
of $L$, $E(Q(e)) = e$ --- that is, the \emph{law of disquotation}
holds universally.

We showed previously that $I = (L, D, V)$ is an interpreted language,
$R= (D_2,S)$ is a syntax representation of $L$, and $(L_2,I)$ is a
syntax language for $R$.  $Q$ is injective since the syntactic
structure of each expression in $L$ is represented by a unique string
in $D_2$.  For $e \in L$, \[V(Q(e)) = V_2(Q(e)) = S(e),\] and thus $Q$
satisfies the Quotation Axiom if $L_{\rm obj} = L$, $D_{\rm syn} =
D_2$, $V_{\rm syn} = S$, and $L_{\rm syn} = L_2$.  For $t \in L_2$
such that $E(t)$ is defined,
\[V(E(t)) = V(V^{-1}_{\rm syn}(V_2(t))) = V(V^{-1}_{\rm syn}(V(t))),\] 
and thus $E$ satisfies the Evaluation Axiom if $L_{\rm obj} = L$,
$D_{\rm syn} = D_2$, $V_{\rm syn} = S$, and $L_{\rm syn} = L_2$.

Therefore, \[F = (D_2,S,L_2,Q,E)\] is a syntax framework for
$(L,I)$. Notice that $E$ is actually the direct evaluation function
for $F$.

\subsection{Example: G\"odel Numbering} \label{subsec:goedel}

Let $L$ be the expressions (i.e., terms and formulas) of a first-order
language of natural number of arithmetic, and let $\sA$ be the
alphabet of symbols occurring in the expressions of $L$.  Once again
the expressions of $L$ can be viewed as strings over the alphabet
$\sA$.  As Kurt G\"odel famously showed in 1931~\cite{Goedel31}, the
syntactic structure of an expression $e$ of $L$ can be represented by
a natural number called the G\"odel number of $e$.  Define $G$ to be
the total function that maps each expression of $L$ to its G\"odel
number.  $G$ is injective since each expression in $L$ has a unique
G\"odel number.  The terms of $L$, which denote natural numbers, can
thus be used to reason \emph{directly} about G\"odel numbers and
\emph{indirectly} about the syntactic structure of the expressions of
$L$.

We will show that this reasoning system based on G\"odel numbers can
be interpreted as a syntax framework.  Let $L_{\rm t}$ be the set of
terms in $L$ and $L_{\rm f}$ be the set of formulas in $L$.  Then \[I=
(L, \mathbb{N} \cup \set{\TRUE,\FALSE}, V),\] where $L = L_{\rm t}
\cup L_{\rm f}$, $\mathbb{N}$ is the set of natural numbers, and $V =
V_{\rm t} \cup V_{\rm f}$, is an interpreted language corresponding to
the language given in Example~\ref{eg:ms-fol}.

Since $G : L \tarrow \mathbb{N}$ is an injective, total function that
maps each expression in $L$ to its G\"odel number, $R =
(\mathbb{N},G)$ is a syntax representation of $L$.  Since $L_{\rm t}
\subseteq L$, $\mathbb{N} \subseteq \mathbb{N} \cup
\set{\TRUE,\FALSE}$, and $V$ restricted to $L_{\rm t}$ is $V_{\rm t} :
L_{\rm t} \tarrow \mathbb{N}$, $(L_{\rm t},I)$ is a syntax language
for $R$.

Let $Q : L \tarrow L_{\rm t}$ be a total function that maps each
expression $e \in L$ to a term $t \in L_{\rm t}$ such that $V_{\rm
  t}(t) = G(e)$.  $Q$ is injective since each expression in $L$ has a
unique G\"odel number.  For $e \in L$, \[V(Q(e)) = V_{\rm t}(Q(e)) =
G(e),\] and thus $Q$ satisfies the Quotation Axiom if $L_{\rm obj} =
L$, $D_{\rm syn} = \mathbb{N}$, $V_{\rm syn} = G$, and $L_{\rm syn} =
L_{\rm t}$.

Let $E : L_{\rm t} \tarrow L$ be the function that, for all $t \in
L_{\rm t}$, $E(t)$ is the expression in $L$ whose G\"odel number is
$V_{\rm t}(t)$ if $V_{\rm t}(t)$ is a G\"odel number of some
expression in $L$ and $E(t)$ is undefined otherwise. For $t \in L_{\rm
  t}$ such that $E(t)$ is defined, \[V(E(t)) = V(G^{-1}(V_{\rm t}(t)))
= V(G^{-1}(V(t))),\] and thus $E$ satisfies the Evaluation Axiom if
$L_{\rm obj} = L$, $D_{\rm syn} = \mathbb{N}$, $V_{\rm syn} = G$, and
$L_{\rm syn} = L_{\rm t}$.  Since not every natural number is a
G\"odel number of an expression in $L$, $G : L \tarrow \mathbb{N}$ is
not surjective and thus $E: L_{\rm t} \tarrow L$ is partial.  For an
expression $e$ of $L$, $Q(e) = t$ such that $V_{\rm t}(t) = G(e)$ by
the definition of $Q$ and then $E(t) = e$ by the the definition of
$E$.  Hence $E(Q(e)) = e$ and so the law of disquotation holds
universally.

Therefore,
\[F = (\mathbb{N},G,L_{\rm t},Q,E)\] 
is a syntax framework for $(L,I)$.  Notice that $E$ is actually the
direct evaluation function for $F$.

Define $L'_{\rm t}$ to be the sublanguage of $L_{\rm t}$ such that $t
\in L'_{\rm t}$ iff $V(t)$ is a G\"odel number of some expression in
$L$.  Then \[F' = (\mathbb{N},G,L'_{\rm t},Q,E'),\] where $E'$ is $E$
restricted to $L'_{\rm t}$, is a syntax framework for $(L,I)$ in which
the evaluation function $E'$ is total.

\subsection{Example: Expressions as Members of an Inductive Type} \label{subsec:ind-type}

In the previous two subsections we saw how strings of symbols and
G\"odel numbers can be used to represent the syntactic structure of
expressions.  These two syntax representations are very popular, but
they are not convenient for practical applications.  In this example
we will see a much more practical syntax representation in which
expressions are represented as members of an inductive type.

Let $L_{\rm prop}$ be a language of propositional logic (with logical
connectives for negation, conjunction, and disjunction). An
interpreter for the language $L_{\rm prop}$ is a program that receives
user input (which we assume is a string), parses the input into a
usable internal representation (i.e., a parse or syntax tree),
computes the value of the internal representation in the form of a new
internal representation, and then displays the new internal
representation in a user-readable form (which we again assume is a
string).  We will describe the components of such an interpreter.

Let \texttt{formula} be the type of the internal data structures
representing the propositional formulas in $L_{\rm prop}$.  This type can be
implemented as an inductive type, e.g., in F\#~\cite{FSharp14}
as:
\begin{verbatim}
type formula =
  | True
  | False
  | Var of string
  | Neg of formula
  | And of (formula * formula)
  | Or  of (formula * formula)
\end{verbatim}
Notice that the type constructors correspond precisely to the various
ways of constructing a well-formed formula in propositional logic.

The interpreter for $L_{\rm prop}$ is the composition of the following functions:

\be

  \item A function \texttt{parse} of type $\mathtt{string} \tarrow
    \mathtt{formula}$ which parses a user input string into an
    internal representation of a well-formed propositional formula ---
    or raises an error if the input does not represent one.  For the
    sake of simplicity, we assume that $L_{\rm prop}$ is chosen so
    that \texttt{parse} is injective.

  \item A function \texttt{value} of type $\mathtt{formula} \tarrow
    \mathtt{formula}$ which determines the truth value of a
    propositional formula of $L_{\rm prop}$ --- or simplifies it in cases that
    contain unknown variables. We will later see how this function
    also requires an additional input $\phi$ of a variable assignment.

  \item A function \texttt{print} of type $\mathtt{formula} \tarrow
    \mathtt{string}$ which prints an internal representation of a
    formula as a string for the user.  We assume that, for each string
    $e$ representing a well-formed propositional formula of $L_{\rm
      prop}$, $\mathtt{print}(\mathtt{parse}(e)) = e$.

\ee

\noindent
For example, suppose $e = \texttt{"p \& true"}$ is a user input string
that denotes a propositional formula in $L_{\rm prop}$.  Then $f =
\texttt{parse}(e) = \texttt{And (Var "p",True)}$ is the expression of
type \texttt{formula} that denotes its internal representation, $f' =
\texttt{value}(f) = \texttt{Var "p"}$ is the expression of type
\texttt{formula} that denotes its computed value, and $e' =
\texttt{print}(f') = \texttt{"p"}$ is the string representation of its
computed value.  Hence the interpretation of $e$
is \[\mathtt{print}(\mathtt{value}(\mathtt{parse}(e))).\] 

%% Moreover, for all expressions $e$ in $L_{\rm prop}$,
%% $\texttt{print}(\mathtt{parse}(e)) = e$.

We will show that this system for interpreting propositional formulas
can be regarded as a syntax framework.  This example demonstrates how
to add a syntax representation and a syntax language to a language
that does not inherently support reasoning about syntax.  It also
demonstrates that any typical implementation of a formal language can
be interpreted as a syntax framework.

Let $L_{\rm prop}$ be the set of well-formed formulas of propositional
logic represented by strings as discussed above, $D_{\rm prop} =
\set{\TRUE,\FALSE}$ be the domain of truth values (i.e., the values
formulas of propositional logic denote), and $V_{\rm prop}^\phi:
L_{\rm prop} \tarrow D_{\rm prop}$ be the semantic valuation function
for propositional logic relative to a variable assignment $\phi$. Then
$I_{\rm prop} = (L_{\rm prop},D_{\rm prop},V_{\rm prop}^\phi)$ is an
interpreted language for propositional logic.

Similarly, let $L_{\rm form}$ be the set of expressions of type
\texttt{formula}, $D_{\rm form}$ be the members of the inductive type
\texttt{formula}, and $V_{\rm form}: L_{\rm form} \tarrow D_{\rm
  form}$ be the semantic valuation function for the expressions of
type \texttt{formula}. Then $I_{\rm form} = (L_{\rm form},D_{\rm
  form},V_{\rm form})$ is also an interpreted language. This secondary
interpreted language is the augmentation that we are adding to the
language of propositional logic in order to represent the syntax of
$L_{\rm prop}$.  Using functions similar to \texttt{parse},
\texttt{value}, and \texttt{print} shown above, we can implement the
language $I_{\rm prop}$ in a programming language.

Let $P: L_{\rm prop} \tarrow D_{\rm form}$ be the function such that,
for $e \in L_{\rm prop}$, $P(e)$ is the value of type
$\mathtt{formula}$ denoted by $\texttt{parse}(e)$.  Then $P$ is an
injective, total function since each $e \in L_{\rm prop}$ has exactly
one parse tree that represents the syntactic structure of $e$.
Therefore, $R = (D_{\rm form},P)$ is a syntax representation of
$L_{\rm prop}$.  The structures $I_{\rm prop}$, $I_{\rm form}$, and
$R$ are depicted in Figure~\ref{fig:impl-lang}.

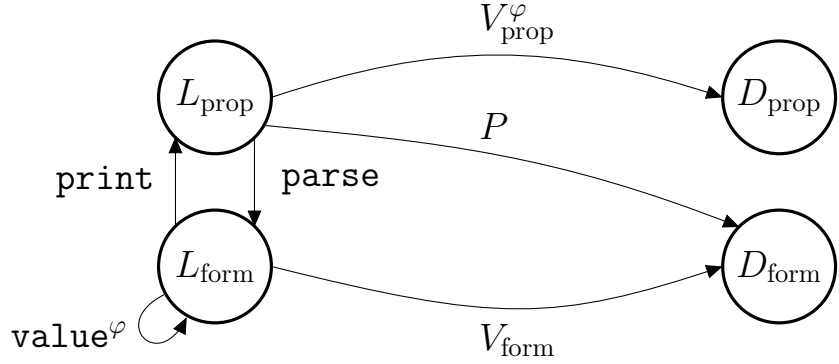
\begin{figure}
\center
\begin{tikzpicture}[scale=.75]
  \draw[very thick] (0,0) circle (1);
    \draw (0,0) node {\Large $L_{\rm prop}$};
  \draw[very thick] (0,-3) circle (1);
    \draw (0,-3) node {\Large $L_{\rm form}$};
  \draw[very thick] (10,0) circle (1);
    \draw (10,0) node {\Large $D_{\rm prop}$};
  \draw[very thick] (10,-3) circle (1);
    \draw (10,-3) node {\Large $D_{\rm form}$};
  \draw[-triangle 45] (1,-3) .. controls (5,-4) and (6,-4) .. (9,-3);
    \draw[right] (4.5,-4.3) node {\Large $V_{\rm form}$};
  \draw[-triangle 45] (1,0) .. controls (4,1) and (6,1) .. (9,0);
    \draw[right] (4.5,1.3) node {\Large $V_{\rm prop}^\phi$};
  \draw[-triangle 45] (0.9,-0.5) .. controls (4,-0.8) and (6, -1) .. (9.3,-2.3);
    \draw[right] (4.5,-0.5) node {\Large $P$};
  \draw[-triangle 45] (0.7,-0.7) -- (0.7,-2.3);
    \draw[right] (1,-1.5) node {\Large $\mathtt{parse}$};
  \draw[-triangle 45] (-0.7,-2.3) -- (-0.7,-0.7);
    \draw[right] (-3,-1.5) node {\Large $\mathtt{print}$};
  \draw[-triangle 45] (-0.9,-3.5) .. controls (-1.9,-4) and (-1,-4.9) .. (-0.5,-3.9);
    \draw[right] (-3.8,-4.2) node {\Large $\mathtt{value}^\phi$};
\end{tikzpicture}
\caption{Domains and Mappings related to $L_{\rm prop}$} \label{fig:impl-lang}
\end{figure}

Let $L=L_{\rm prop} \cup L_{\rm form}$, $D=D_{\rm prop} \cup D_{\rm
  form}$, and $V^\phi = V_{\rm prop}^\phi \cup V_{\rm form}$. $V^\phi$
is a function since the two functions $V_{\rm prop}^\phi$ and $V_{\rm
  form}$ have disjoint domains. Then $I=(L,D,V^\phi)$ is an
interpreted language and $(L_{\rm form},I)$ is a syntax language for
$R$ by construction.

The tuple \[F=(D_{\rm form},P,L_{\rm
  form},\mathtt{parse},\mathtt{print})\] is a syntax framework for
$(L_{\rm prop},I)$ since: 

\be

  \item $R=(D_{\rm form},P)$ is a syntax representation of $L_{\rm
    prop}$ as shown above.

  \item $(L_{\rm form},I)$ is syntax language for $R$ as shown above.

  \item \textbf{Quotation Axiom}: For all $e \in L_{\rm prop}$,
    $P(e)=V_{\rm form}(\mathtt{parse}(e))$ by definition, and thus
    \[V^\phi(\mathtt{parse}(e)) = V_{\rm form}(\mathtt{parse}(e)) = P(e).\]

  \item \textbf{Evaluation Axiom}: For all $e \in L_{\rm form}$,
    $P^{-1}(V_{\rm form}(e)) = \mathtt{parse}^{-1}(e) =
    \mathtt{print}(e)$ since $\mathtt{print}(\mathtt{parse}(e)) = e$,
    and thus
    \[V^\phi(\mathtt{print}(e)) = V^\phi(P^{-1}(V_{\rm form}(e))) = V^\phi(P^{-1}(V^\phi(e))).\]
\ee
\noindent
Since $\texttt{print}(\mathtt{parse}(e)) = e$ holds for all
expressions in $L_{\rm prop}$, the Law of Disquotation holds
universally.

The syntax framework for this example provides the structure that is
needed to understand the function $\mathtt{value}^\phi$ shown in
Figure~\ref{fig:impl-lang} as an implementation of the semantic
valuation function $V_{\rm prop}^\phi$. The formula that specifies
$\mathtt{value}^\phi$, \[V^\phi(e) =
V^\phi(\mathtt{print}(\mathtt{value}^\phi(\mathtt{parse}(e)))),\]
illustrates the interplay of syntax and semantics that is inherent in
its meaning.

The approach employed in this third example, in which the syntactic
values are members of an inductive type, is commonly used in
programming to represent syntax (see~\cite{FriedmanWand08}).  It
utilizes a \emph{deep embedding}~\cite{BoultonEtAl93} of the object
language $L_{\rm obj}$ into the full underlying formal language $L$.

\subsection{Further Remarks}

\begin{rem}[Variable Binding]\em
None of the standard examples discussed above treat variable binding
constructions in any special way.  There are other syntax
representation methods that identify expressions that are the same up
to a renaming of the variables that are bound by variable binders.
One method is \emph{higher-order abstract
  syntax}~\cite{Miller00,PfenningElliot88} in which the syntactic
structure of an expression with variable binders is represented by a
term in typed lambda calculus.  Another method is \emph{nominal
  techniques}~\cite{GabbayPitts02,Pitts03} in which the swapping of
variable names can be explicitly expressed.  The
paper~\cite{NanevskiPfenning05} combines quotation/evaluation
techniques with nominal techniques.\hfill $\Box$
\end{rem}

\begin{rem}[Types]\em
The languages in a syntax framework are not required to be typed.
However, it is natural that, if an expression $e$ in the object
language is of type $\alpha$, then $Q(e)$ should be of some type
$\mname{expr}(\alpha)$.  The operator $\mname{expr}$ behaves like the
necessity operator $\Box$ in modal logic~\cite{DaviesPfenning01}.  An
important design decision for such a type system is whether or not
every expression of the syntax language equals a quotation of an
expression.  In other words, should a syntax framework with a type
system admit only expressions in the syntax language that denote the
syntactic structure of well-formed expressions or should it admit in
addition expressions that denote the syntactic structure of ill-formed
expressions.  Recall that in the example of
subsection~\ref{subsec:goedel} the syntax language of $F$ contains the
latter kind of expressions, while the syntax language of $F'$ contains
only the former kind.\hfill$\Box$
\end{rem}

\section{Syntax Frameworks with Built-In Operators} \label{sec:built-in}

The three examples in the previous section illustrate how a syntax
framework provides the means to reason about the syntax of a
designated object language $L_{\rm obj} \subseteq L$.  In all three
examples, only \emph{indirect statements} about the syntax of $L_{\rm
  obj}$ can be expressed in $L$, while direct statements using $Q$ and
$E$ can be expressed in the metalanguage of $L$.  In this section we
will explore syntax frameworks in which \emph{direct statements} about
the syntax of $L_{\rm obj}$, such as $E(Q(e)) = e$, can be expressed
in $L$ itself.

\subsection{Built-in Quotation and Evaluation}

\bsp Let $I=(L,D,V)$ be an interpreted language, $L_{\rm obj}$ be a
sublanguage of $L$, and $F=(D_{\rm syn}, V_{\rm syn}, L_{\rm syn}, Q,
E)$ be a syntax framework for $(L_{\rm obj},I)$.  $F$ has
\emph{built-in quotation} if there is an operator (which we will
denote as \mname{quote}) such that, for all $e \in L_{\rm obj}$,
$Q(e)$ is the syntactic result of applying the operator to $e$ (which
we will denote as $\mname{quote}(e)$).  $F$ has \emph{built-in
  evaluation} if there is an operator (which we will denote as
\mname{eval}) such that, for all $e \in L_{\rm syn}$, $E(e)$ is the
syntactic result of applying the operator to $e$ (which we will denote
as $\mname{eval}(e)$) whenever $E(e)$ is defined.\footnote{If $L_{\rm
    obj}$ is a typed language, it may be necessary for the
  \mname{eval} operator to include a parameter that ranges over the
  types of the expressions in $L_{\rm obj}$.}  There are similar
definitions of built-in quotation and evaluation for syntax frameworks
in interpreted theories. \esp

Assume $F$ has both built-in quotation and evaluation.  Then
quotations and evaluations are expressions in $L$, and $F$ thus
provides the means to reason directly in $L$ about the interplay of
the syntax and semantics of the expressions in $L_{\rm obj}$.  In
particular, it is possible to specify in $L$ the semantic meanings of
transformers.  The following lemma shows that, since the quotations
and evaluations in $F$ begin with the operators \mname{quote} and
\mname{eval}, respectively, $E$ cannot be the direct evaluation for
$F$.  

\begin{lem}
Suppose $F$ is a syntax framework that has built-in quotation and
evaluation.  Then $E \not= E^{\ast}$.
\end{lem}

\begin{proof}
Suppose $E = E^{\ast}$.  Let $e \in L_{\rm obj}$.  Then
\setcounter{equation}{0}
\begin{eqnarray}
e & = & V_{\rm syn}^{-1}(V_{\rm syn}(e)) \\
  & = & V_{\rm syn}^{-1}(V_{\rm sem}(\mname{quote}(e))) \\
  & = & E^{\ast}(\mname{quote}(e)) \\
  & = & E(\mname{quote}(e)) \\
  & = & \mname{eval}(\mname{quote}(e)) 
\end{eqnarray}
(1) is by the fact that $V_{\rm syn}$ is total on $L_{\rm obj}$; (2)
is by built-in quotation and the Quotation Axiom; (3) is by the
definition of the direct evaluation function; (4) is by hypothesis;
and (5) is by the fact that $E$ is built in.  Hence $e =
\mname{eval}(\mname{quote}(e))$, which is a contradiction since these
are syntactically distinct expressions.
\end{proof}

\bigskip

The syntax framework $F$ is \emph{replete} if the object language of
$F$ is equal to the full language of $F$ (i.e., $L_{\rm obj} = L$) and
$F$ has both built-in quotation and evaluation.  A replete syntax
framework whose full language is $L$ has the facility to reason about
the syntax of all of $L$ within $L$ itself.  $F$ is \emph{weakly
  replete} if $L_{\rm syn} \subseteq L_{\rm obj}$ and $F$ has both
built-in quotation and evaluation.  There are similar definitions of
replete and weakly replete for syntax frameworks in interpreted
theories.  We will give two examples of a replete syntax framework,
one in the next subsection and one in section~\ref{sec:literature}.
We will also give another example in section~\ref{sec:literature} of a
syntax framework that is almost replete.

\begin{rem}\em
A \emph{biform
  theory}~\cite{CaretteFarmer08,Farmer07b,FarmerMohrenschildt03} is a
combination of an axiomatic theory and an algorithmic theory.  It is a
basic unit of mathematical knowledge that consists of a set of
\emph{concepts}, \emph{transformers}, and \emph{facts}.  The concepts
are symbols that denote mathematical values and, together with the
transformers, form a language $L$ for the theory.  The transformers
are programs whose input and output are expressions in $L$; they
represent syntax-based algorithms like reasoning rules.  The facts are
statements expressed in $L$ about the concepts and transformers.  A
logic with a replete syntax framework (such as Chiron discussed in
subsection~\ref{subsec:chiron}) is well-suited for formalizing biform
theories~\cite{Farmer07b}.\hfill $\Box$
\end{rem}

\subsection{Example: Lisp} \label{subsec:lisp}

We will show that the Lisp programming language with a simplified
semantics is an instance of a syntax framework with built-in quotation
and evaluation.

Choose some standard implementation of Lisp.  Let $L$ be the set of
S-expressions that do not change the Lisp valuation context when they
are evaluated by the Lisp interpreter.  Let $V : L \tarrow L \cup
\set{\bot}$ be the total function that, for all S-expressions $e \in
L$, $V(e)$ is the S-expression the interpreter returns when $e$ is
evaluated if the interpreter returns an S-expression in $L$ and $V(e)
= \bot$ otherwise.  $I = (L,L \cup \set{\bot},V)$ is thus an
interpreted language.

$R = (L,\mname{id}_L)$, where $\mname{id}_L$ is the identity function
on $L$, is a syntax representation of $L$ since each S-expression
represents its own syntactic structure.  Let $L'$ be the sublanguage
of $L$ such that, for all $e \in L$, $e \in L'$ iff $V(e) \not= \bot$.
It follows immediately by the definition of $L'$ that $(L',I)$ is a
syntax language for $R$.

Let $Q : L \tarrow L'$ be the total function that maps each $e \in L$
to the S-expression $(\texttt{quote} \; e)$.  For $e \in L$, $Q(e) \in
L'$ since $V((\texttt{quote} \; e)) = e \not= \bot$.  $Q$ is obviously
injective.  For $e \in L$, \[V(Q(e)) = V((\texttt{quote} \; e)) = e =
\mname{id}_L(e),\] and thus $Q$ satisfies the Quotation Axiom if
$L_{\rm obj} = L$, $D_{\rm syn} = L$, $V_{\rm syn} = \mname{id}_L$,
and $L_{\rm syn} = L'$.

Let $E : L' \tarrow L$ be the total function that, for all $e \in L'$,
$E(e)$ is the S-expression $(\texttt{eval} \; e)$.  For all $e \in
L'$, \[V(E(e)) = V((\texttt{eval} \; e)) = V(V(e)) =
V(\mname{id}^{-1}_L(V(e))).\] (Notice that $V(V(e))$ is always defined
since $e \in L'$.)  Thus $E$ satisfies the Evaluation Axiom if $L_{\rm
  obj} = L$, $D_{\rm syn} = L$, $V_{\rm syn} = \mname{id}_L$, and
$L_{\rm syn} = L'$.

Therefore, \[F = (L,\mname{id}_L,L',Q,E)\] is a replete syntax
framework for $(L,I)$.

Suppose $L$ were the full set of S-expressions, including the
S-expressions that modify the Lisp valuation context when they are
evaluated by the interpreter.  Then, in order to interpret Lisp as a
syntax framework, we would need to extend the notion of a \emph{syntax
  framework} to the notion of \emph{contextual syntax framework} as
mentioned in Remark~\ref{rem:contextual}

\subsection{Example: Liar Paradox} \label{subsec:liar}

The virtue of a syntax framework with built-in quotation and
evaluation is that it provides the means to express statements about
the interplay of the syntax and semantics of the expressions in
$L_{\rm obj}$ in $L$.  On the other hand, the vice of such a syntax
framework is that, if $L$ is sufficiently expressive, the liar paradox
can be expressed in $L$ using quotation and evaluation.

\bsp Let $I= (L, \mathbb{N} \cup \set{\TRUE,\FALSE}, V)$ be the
interpreted language and $F' = (\mathbb{N},G,L'_{\rm t},Q,E')$ be the
syntax framework for $(L,I)$ given in subsection~\ref{subsec:goedel}.  Assume
that $V$ is defined so that the axioms of first-order Peano arithmetic
are satisfied (see~\cite{Mendelson09}).  Assume also that $F'$ has
been modified so that it has both built-in quotation and built-in
evaluation. \esp

\bsp We claim $E'$ cannot be total.  Assume otherwise.  By the
diagonalization lemma~\cite{Carnap34}, there is an expression $A \in
L$, such that $V(A) = V(\mname{quote}(\Neg(\mname{eval}(A))))$.  Then
\esp \setcounter{equation}{0}
\begin{eqnarray}
V(\mname{eval}(A)) 
  & = & V(G^{-1}(V(A))) \\
  & = & V(G^{-1}(V(\mname{quote}(\Neg(\mname{eval}(A)))) \\
  & = & V(G^{-1}(G(\Neg(\mname{eval}(A))))) \\
  & = & V(\Neg(\mname{eval}(A)))
\end{eqnarray}
(1) is by built-in evaluation, the totality of $E'$, and the
Evaluation Axiom; (2) is by the definition of $A$; (3) is by built-in
quotation and the Quotation Axiom, and (4) is by the fact $G$ is total
on $L$.  Hence $V(\mname{eval}(A)) = V(\Neg(\mname{eval}(A)))$, which
contradicts the fact that $V$ never assigns a formula and its negation
the same truth value.  Therefore, $E'$ cannot be total and, in
particular, cannot be total on quotations.

The formula $\mname{eval}(A)$ expresses the \emph{liar paradox} and
the argument above is a proof of Alfred Tarski's 1933 theorem on the
undefinability of truth~\cite{Tarski33,Tarski35,Tarski35a}, which says
that built-in evaluation cannot serve as a truth predicate over all
formulas.  This example demonstrates why evaluation is allowed to be
partial in a syntax framework: if evaluation were required to be
total, the notion of a syntax framework would not cover reasoning
systems with built-in quotation and evaluation in which the liar
paradox can be expressed.

\subsection{Example: G\"odel Numbering with Built-In Quotation}

A syntax framework without built-in quotation and evaluation can
sometimes be modified to have built-in quotation or evaluation.

\bsp Let $I = (L, \mathbb{N} \cup \set{\TRUE,\FALSE}, V)$ be the
interpreted language and $F' = (\mathbb{N},G,L'_{\rm t},Q,E')$ be the
syntax framework for $(L,I)$ given in subsection~\ref{subsec:goedel}.
Extend $L$ to the language $L^{\ast}$ and $L'_{\rm t}$ to
$L^{\ast}_{\rm t}$ by adding a new operator $\mname{quote}$ so that
$\mname{quote}(e) \in L^{\ast}_{\rm t}$ for all $e \in L^{\ast}$.
Extend $G$ to $G^{\ast}: L^{\ast} \tarrow \mathbb{N}$ so that
$G^{\ast}(e)$ is the G\"odel number of $e$ for all $e \in L^{\ast}$.
Extend $V$ to $V^{\ast} : L^{\ast} \tarrow \mathbb{N} \cup
\set{\TRUE,\FALSE}$ so that $V^{\ast}(\mname{quote}(e)) = G^{\ast}(e)$
for all $e \in L^{\ast}$.  And, finally, define $Q^{\ast}(e)$ to be
$\mname{quote}(e)$ for all $e \in L^{\ast}$.  (We do not need to
change the definition of $E'$.)  Then $I^{\ast} = (L^{\ast},
\mathbb{N} \cup \set{\TRUE,\FALSE}, V^{\ast})$ is an interpreted
language and \[F^{\ast} = (\mathbb{N},G^{\ast},L^{\ast}_{\rm
  t},Q^{\ast},E')\] is a syntax framework for $(L^{\ast},I^{\ast})$
that has built-in quotation. \esp

See \cite{Farmer13} for further discussion on the challenges involved
in modifying a traditional logic to embody the structure of a replete
syntax framework.

\section{Quasiquotation} \label{sec:quasiquotation}

Quasiquotation is a parameterized form of quotation in which the
parameters serve as holes in a quotation that are filled with the
values of expressions.  It is a very powerful syntactic device for
specifying expressions and defining macros.  Quasiquotation was
introduced by Willard Quine in 1940 in the first version of his book
\emph{Mathematical Logic}~\cite{Quine03}.  It has been extensively
employed in the Lisp family of programming
languages~\cite{Bawden99}.\footnote{In Lisp, the standard symbol for
  quasiquotation is the backquote (\texttt{`}) symbol, and thus in
  Lisp, quasiquotation is usually called \emph{backquote}.}

\bsp We will show in this section how quasiquotation can be defined in
a syntax framework.  Let $I=(L,D,V)$ be an interpreted language, $L_{\rm
  obj}$ be a sublanguage of $L$, and $F=(D_{\rm syn}, V_{\rm syn},
L_{\rm syn}, Q, E)$ be a syntax framework for $(L_{\rm obj},I)$. \esp

\subsection{Marked Expressions} \label{subsec:marked-expr}

Suppose $e \in L$.  A \emph{subexpression} of $e$ is an occurrence in
$e$ of some $e' \in L$.  We assume that there is a set of
\emph{positions} in the syntactic structure of $e$ such that each
subexpression of $e$ is indicated by a unique position in $e$.  Two
subexpressions $e_1$ and $e_2$ of $e$ are \emph{disjoint} if $e_1$ and
$e_2$ do not share any part of the syntactic structure of $e$.

Let $e \in L_{\rm obj}$.  A \emph{marked expression} derived from $e$
is an expression of the form $e\mlist{(p_1,e_1),\ldots,(p_n,e_n)}$
where $n \ge 0$, $p_1,\ldots,p_n$ are positions of pairwise disjoint
subexpressions of $e$, and $e_1,\ldots,e_n$ are expressions in $L$.
Define $L^{\rm m}_{\rm obj}$ to be the set of marked expressions
derived from members of $L_{\rm obj}$.

\bsp Let $S: L^{\rm m}_{\rm obj} \tarrow L_{\rm obj}$ be the function
that, given a marked expression $m =
e\mlist{(p_1,e_1),\ldots,(p_n,e_n)} \in L^{\rm m}_{\rm obj}$,
simultaneously replaces each subexpression in $e$ at position $p_i$
with $E^{\ast}(e_i)$ (the application of the direct evaluation
function for $F$ to $e_i$) for all $i$ with $1 \le i \le n$.  $S(e)$
will be undefined if either $E^{\ast}(e_i)$ is undefined or
$E^{\ast}(e_i)$ does not have the same type as the subexpression at
position $p_i$ for some $i$ with $1 \le i \le n$.  \esp

\subsection{Quasiquotation}

Define $\overline{Q} : L^{\rm m}_{\rm obj} \tarrow L_{\rm syn}$ to be
the (possibly partial) function such that, if $m =
e\mlist{(p_1,e_1),\ldots,(p_n,e_n)} \in L^{\rm m}_{\rm obj}$, then
$\overline{Q}(m) = Q(S(m))$.  $\overline{Q}(m)$ is defined iff $S(m)$
is defined.  For $m \in L^{\rm m}_{\rm obj}$, $\overline{Q}(m)$ is
called the \emph{quasiquotation} of $m$.\footnote{The
  position-expression pairs $(p_i,e_i)$ in a quasiquotation
  $\overline{Q}(e\mlist{(p_1,e_1),\ldots,(p_n,e_n)})$ are sometimes
  called \emph{antiquotations}.}

$F$ has \emph{built-in quasiquotation} if there is an operator (which
we will denote as \mname{quasiquote}) such that, for all $m =
e\mlist{(p_1,e_1),\ldots,(p_n,e_n)} \in L^{\rm m}$, $\overline{Q}(m)$
is the syntactic result of applying the operator to
$e,p_1,\ldots,p_n,e_1,\ldots,e_n$ (which we will denote as
$\mname{quasiquote}(m)$).

\subsection{Backquote in Lisp} \label{subsec:backquote} 

Let us continue the example in subsection~\ref{subsec:lisp} involving
Lisp with a simplified semantics.  In Lisp, a \emph{backquote} of $L$
is an expression of the form $\texttt{`}e$ where $e$ is an
S-expression in $L$ in which some of the subexpressions of $e$ are
marked by a comma (\texttt{,}).  For example,
\[\texttt{`(+ 2 ,(+ 3 1))}\] is a backquote in which \texttt{(+ 3 1)} 
is a subexpression marked by a comma.  We will restrict our attention
to unnested backquotes.  The Lisp interpreter normally returns an
S-expression when it evaluates a backquote $\texttt{`}e \in L$.  In
this case the S-expression returned is obtained from $e$ by replacing
each subexpression $e'$ in $e$ marked by a comma with the S-expression
$V(e')$.  For example, when evaluating \texttt{`(+ 2 ,(+ 3 1))}, the
interpreter returns \mbox{\texttt{(+ 2 4)}}.  Let $L$ be extended to
$L^\ast$ to include the backquotes of $L$ and $V^\ast : L^\ast \tarrow
L^\ast \cup \set{\bot}$ be the total function such that, for all
S-expressions and backquotes $e \in L^\ast$, $V^\ast(e)$ is the
S-expression the interpreter returns when $e$ is evaluated if the
interpreter returns an S-expression and $V^\ast(e) = \bot$ otherwise.

\bsp A backquote $\texttt{`}e$ in $L^\ast$ corresponds to a marked
expression $m = e\mlist{(p_1,e_1),\ldots,(p_n,e_n)} \in L^{\rm m}$
where each $p_i$ is the position of a subexpression $\texttt{,}e_i$ in
$e$ marked by a comma for all $i$ with $1 \le i \le n$.  Let
$\texttt{`}e \in L^\ast$ be a backquote and $m =
e\mlist{(p_1,e_1),\ldots,(p_n,e_n)} \in L^{\rm m}$ be a marked
expression that corresponds to it.  We will show that the semantic
value of the backquote $\texttt{`}e$, when it is not $\bot$, is the
same as the semantic value of the quasiquotation $\overline{Q}(m)$.
Assume $V^\ast(\texttt{`}e) \not= \bot$.  Then
\setcounter{equation}{0}
\begin{eqnarray}
V^\ast(\texttt{`}e) 
  & = & S(m) \\
  & = & V(Q(S(m))) \\
  & = & V(\overline{Q}(m))
\end{eqnarray}
(1) is by the semantics of backquote and the definition of $S$ since
\[V(e_i) = \mname{id}_{L}^{-1}(V(e_i)) = V_{\rm syn}^{-1}(V(e_i)) =
E^{\ast}(e_i)\] for each $i$ with $1 \le i \le n$.  (2) is by the
Quotation Axiom and the fact that $V_{\rm syn}$ is the identity
function.  And (3) is by the definition of $\overline{Q}(m)$. \esp

\section{Examples from the Literature} \label{sec:literature}

\subsection{Example: Lambda Calculus} \label{subsec:lambda}
\newcommand{\betaarrow}{\twoheadrightarrow_{\beta}}
\newcommand{\nflambda}{{\rm NF}_{\Lambda}}
\newcommand{\nfdomain}{\nflambda \cup \{\bot\}}

In 1994 Torben Mogensen~\cite{Mogensen94} introduced a method of self
representing and interpreting terms of lambda calculus. We will
analyze this method and demonstrate how the self-interpretation of
lambda calculus is almost an instance of a replete syntax framework.

Let $\Lambda = V ~|~ \Lambda ~ \Lambda ~|~ \lambda V \mdot \Lambda$ be
the set of $\lambda$-terms where $V$ is a countable set of
variables. $\Lambda$ is the language of lambda calculus consisting of
all the $\lambda$-terms.  A $\lambda$-term is a \emph{normal form} if
$\beta$-reduction cannot be applied to it.  Given a $\lambda$-term
$M$, let the \emph{normal form of $M$}, ${\rm NF}_M$, be the normal
form that results from repeatedly applying $\beta$-reduction to $M$
until a normal form is obtained.  The normal form of $M$ is undefined
if a normal form is never obtained after repeatedly applying
$\beta$-reduction to $M$.  We will introduce two different syntax
representations of this language. The first syntax representation of
$\Lambda$ uses an inductive type similar to
subsection~\ref{subsec:ind-type} such that $V_A$ is the syntactic
valuation function where: \setcounter{equation}{0}
\begin{eqnarray}
V_A(x) & = & {\tt Var}(x)\\
V_A(M~N) & = & {\tt App}(V_A(M), V_A(N))\\
V_A(\lambda x \mdot M) & = & {\tt Abs}(\lambda x \mdot V_A(M))
\end{eqnarray}
\noindent 
Let $D_A$ be the domain of values of this inductive type.
Then $R_A = (D_A,V_A)$ is a syntax representation of $\Lambda$.

Mogensen~\cite{Mogensen94} suggests a different syntax representation
of lambda calculus. Let $\bsynbrack{\cdot} : \Lambda \tarrow
\nflambda$ be a {\em representation schema} for lambda calculus such
that: 
\setcounter{equation}{0}
\begin{eqnarray}
\bsynbrack{x} & = & \lambda a b c \mdot a ~ x\\
\bsynbrack{M~N} & = & \lambda a b c \mdot b ~ \bsynbrack{M} ~ \bsynbrack{N}\\
\bsynbrack{\lambda x \mdot M} & = & \lambda a b c \mdot c ~ (\lambda x \mdot \bsynbrack{M})
\end{eqnarray}
\noindent 
where $a,b,c$ are variables not occurring free in the $\lambda$-terms
$M$ and $N$. This representation of $\lambda$-terms is an equivalent
representation to the method described earlier which utilizes the
constructs of lambda calculus itself instead of an external data type.

Then $R_{\Lambda} = (\nflambda, \bsynbrack{\cdot})$ is a syntax
representation of $\Lambda$ and $(\nflambda, I_{\Lambda})$ is a
syntax language for $R_{\Lambda}$. Notice that, since $\bsynbrack{M}$
is in normal form for any $M \in \Lambda$, then trivially
$\bsynbrack{M} \betaarrow \bsynbrack{M}$.

Let a {\em self-interpreter} $E$ be a $\lambda$-term such that for any
$M \in {\Lambda}$, $E \bsynbrack{M}$ is $\beta$-equivalent to $M$,
i.e., $E \bsynbrack{M} =_{\beta} M$ (which means ${\rm NF}_{E
  \bsynbrack{M}}$ and ${\rm NF}_M$ are $\alpha$-convertible when these
normal forms exist).  Mogensen proves that the $\lambda$-term
\[E = Y ~ \lambda e \mdot \lambda m \mdot m ~ (\lambda x \mdot x) ~ (\lambda m n \mdot (e~m)~(e~n)) ~ (\lambda m \mdot \lambda v \mdot e(m~v)),\]
\noindent where $Y$ is the Y-combinator, is a self-interpreter.
Define $E_{\Lambda} : \nflambda \tarrow \Lambda$ to be the partial
function such that $E_{\Lambda}(M) = E~M$ if $M = \bsynbrack{N_M}$ for
some $\lambda$-term $N_M$ and is undefined otherwise.

\iffalse
The immediate implementation of $E$ for the syntax representation $R_A$ is as follows:
\setcounter{equation}{0}
\begin{eqnarray}
E_A[{\tt Var}(x)] & = & x \\
E_A[{\tt App}(V_A(M), V_A(N))] & = & E_A[M] ~ E_A[N]\\
E_A[{\tt Abs}(\lambda x \mdot V_A(M))] & = & \lambda v \mdot E_A[M ~ v]
\end{eqnarray}
\fi

\begin{thm}\bsp
Let $\Lambda$ be the language of lambda calculus and $I_{\Lambda} =
(\Lambda, \nfdomain, \betaarrow)$ be the interpreted language of
lambda calculus as defined earlier. Let $\bsynbrack{\cdot}$ be the
representation schema of $\Lambda$ and $E_{\Lambda}$ be the function
defined above.  Then \[F_{\Lambda} = (\nflambda, \bsynbrack{\cdot},
\nflambda, \bsynbrack{\cdot}, E_{\Lambda})\] is a syntax framework for
$(\Lambda,I_{\Lambda})$. \esp
\end{thm}

\begin{proof}
$F_{\Lambda}$ is a syntax framework since it satisfies the four
  conditions of Definition~\ref{df:syn-frame-lang}:

\be

  \item $R_{\Lambda} = (\nflambda, \bsynbrack{\cdot})$ is a syntax
    representation of $\Lambda$.

  \item $(\nflambda, I_{\Lambda})$ is syntax language for
    $R_{\Lambda}$.

  \item $\bsynbrack{\cdot} : \Lambda \tarrow \nflambda$ is an
    injective, total function such that, for all $M \in \Lambda$,
    $\bsynbrack{M} \betaarrow \bsynbrack{M}$ (Quotation Axiom).

  \item $E_{\Lambda} : \nflambda \tarrow \Lambda$ is a partial
    function such that, for all $M \in \nflambda$ with $M =
    \bsynbrack{N_M}$ for some $\lambda$-term $N_M$, $E_{\Lambda}(M) =
    E~M = E\bsynbrack{N_M} =_{\beta} N_M$ (Evaluation Axiom) since $E$
    is a self-interpreter.

\ee
\end{proof}

\bigskip

$F_{\Lambda}$ is almost replete: $\Lambda$ is both the object and full
language of $F_{\Lambda}$ and $F_{\Lambda}$ has built-in evaluation,
but $F_{\Lambda}$ does not have built-in quotation.

\iffalse
$F_{\Lambda}$ does not have built-in quotation, but this can be
defined in $F_{\Lambda}$ by introducing a constant $C$ such that $C~M
= \bsynbrack{M}$ for all $M \in \Lambda$.  Hence the Mogensen
self-interpretation of lambda calculus can be formulated as a replete
syntax framework.
\fi

\iffalse
\bsp Mogensen also introduces a {\em self-reducer} $R$ for
$\lambda$-terms such that $R ~ \bsynbrack{M} =_{\beta} \bsynbrack{{\rm
    NF}_M}$ and provides a proof of correctness for the
self-reducer. Let $M \in \Lambda$ be a $\lambda$-term and
$\overline{M} \in {\Lambda}^m$ be the marked expression
$M\mlist{(p,\bsynbrack{M})}$, where $p$ is the top position in $M$, as
in section~\ref{subsec:marked-expr}.  Then $\overline{Q}(\overline{M})
= \bsynbrack{E^{\ast}(\bsynbrack{M})} = \bsynbrack{{\rm
    NF}_M}$. Therefore, $R ~ \bsynbrack{M} =_{\beta}
\overline{Q}(M\mlist{(p,\bsynbrack{M})})$ and the self-reducer for
lambda calculus is a built-in special form of the quasiquotation in
syntax frameworks.  \esp
\fi

\subsection{Example: The Ring Tactic in Coq} \label{subsec:coq}
Coq~\cite{Coq8.4} is an interactive theorem prover based on the
calculus of inductive constructions. Let $R$ be a ring with the
associative, commutative binary operators $+$ and $*$ and the
constants $0$ and $1$ that are the identities of $+$ and $*$,
respectively. A {\em polynomial} in $R$ is an expression that consists
of the constants of $R$, the operators $+$ and $*$, and variables
$v_0, v_1, \dots$ of type $R$.

The {\em ring tactic} in Coq is a polynomial simplifier that converts
any polynomial to its equivalent {\em normal form}. The normal form of
a polynomial is defined as the ordered sum of unique monomials in
lexicographic order.

Earlier we mentioned that syntax-based operations such as
(symbolically) computing derivatives require a syntax framework to
manipulate and reason about syntax using quotation and evaluation.
Polynomial simplification is a term rewriter that uses the quotation
and evaluation mechanisms.  The {\tt ring} tactic in Coq automatically
quotes and simplifies every polynomial expression.

Internally, when the {\tt ring} tactic is applied, the polynomials are
represented by an inductive type {\tt polynomial}. The Coq
reference manual~\cite{Coq8.4} defines this type as:
\begin{verbatim}
Inductive polynomial : Type :=
  | Pvar : index -> polynomial
  | Pconst : A -> polynomial
  | Pplus : polynomial -> polynomial -> polynomial
  | Pmult : polynomial -> polynomial -> polynomial
  | Popp : polynomial -> polynomial.
\end{verbatim}
which represents polynomials similar to the inductive type example in
subsection \ref{subsec:ind-type}.

Let $L$ be the language of Coq, $D$ be the semantic domain of values
in the calculus of inductive constructions, and $V$ be the semantic
interpreter of Coq, then $I = (L,D,V)$ is the interpreted language for
Coq. Let $L_R \subseteq L$ be the language of polynomials of type $R$
(i.e., expressions in $L$ that are built with operators and constants
of $R$ and variables $v_0, v_1, \dots$ as defined earlier), $L_{\rm
  poly} \subseteq L$ be the language of expressions belonging to the
inductive type {\tt polynomial}, $D_{\rm poly} \subseteq D$ be the
image of $L_{\rm poly}$ under $V$, and $V_{\rm poly}$ be the internal
quotation mechanism of Coq the {\tt ring} tactic uses to lift
polynomial expressions in $L_R$ to expressions in $L_{\rm poly}$.
Then $(D_{\rm poly},V_{\rm poly})$ is a syntax representation and
$(L_{\rm poly},I)$ is a syntax language for this syntax representation
which is suitable for describing the \texttt{ring} tactic in Coq.

Coq's ring normalization library ({\tt Ring\_normalize.v}) also defines
an interpretation function that transforms a polynomial expression of
type {\tt polynomial} back to a ring value of type $R$:
\begin{verbatim}
Fixpoint interp_p (p:polynomial) : A :=
  match p with
  | Pconst c => c
  | Pvar i => varmap_find Azero i vm
  | Pplus p1 p2 => Aplus (interp_p p1) (interp_p p2)
  | Pmult p1 p2 => Amult (interp_p p1) (interp_p p2)
  | Popp p1 => Aopp (interp_p p1)
  end.
\end{verbatim}

To finish a definition of a syntax framework for the \texttt{ring}
tactic in Coq, we need to construct two functions $Q: L_R \to L_{\rm
  poly}$ and $E : L_{\rm poly} \to L_R$ in the metalanguage of Coq.
Their definitions are:

\be

  \item For all $e \in L_R$, $Q(e)$ is the $e' \in L_{\rm poly}$ such
    that $V(e') = V_{\rm poly}(e).$

  \item \bsp For all $e' \in L_{\rm poly}$, $E(e')$ is the $e \in L_R$ such
    that $V(e) = V(\mathtt{interp\_p})(V(e'))$. \esp

\ee
Then $F = (D_{\rm poly},V_{\rm poly},L_{\rm poly},Q,E)$ is a syntax
framework for $(L_R,I)$.

Notice that the two functions $Q$ and $E$ are not normally present in
Coq and were constructed by the machinery in Coq described above
specifically to satisfy the requirements of a syntax framework.
Although the concepts of the syntax language and the syntax
representation arose naturally from the internal mechanism for the
\texttt{ring} tactic in Coq, a syntax framework for the \texttt{ring}
tactic does \emph{not} reside in Coq as explicitly as our previous
examples.

\subsection{Example: Chiron} \label{subsec:chiron}

Chiron~\cite{Farmer07a,Farmer12}, is a derivative of
von-Neumann-Bernays-G\"odel ({\nbg}) set theory~\cite{Goedel40,
  Mendelson09} that is intended to be a practical, general-purpose
logic for mechanizing mathematics.  Unlike traditional set theories
such as Zermelo-Fraenkel ({\zf}) and {\nbg}, Chiron is equipped with a
type system, and unlike traditional logics such as first-order logic
and simple type theory, Chiron admits undefined terms.  The most
noteworthy part of Chiron is its facility for reasoning about the
syntax of expressions that includes built-in quotation and evaluation.

We will assume that the reader is familiar with the definitions
concerning Chiron in~\cite{Farmer12}.  Let $L$ be a language of
Chiron, $\sE_L$ be the set of expressions in $L$, $M$ be a standard
model for $L$, $D_M$ be the set of values in $M$, $V$ be the valuation
function in $M$, and $\phi$ be an assignment into $M$.  Then
$I=(\sE_L,D_M,V_\phi)$ is an interpreted language.

$D_M$ includes certain sets called \emph{constructions} that are
isomorphic to the syntactic structures of the expressions in $\sE_L$.
$H$ is a function in $M$ that maps each expression in $\sE_L$ to a
construction representing it.  Let $D_{\rm syn}$ be the range of $H$
and $\sT_{\rm syn}$ be the set of terms $a$ such that $V_\phi(a) \in
D_{\rm syn}$.  For $e \in \sE_L$, define $Q(e) = (\mname{quote},e)$.
For $a \in \sT_{\rm syn}$, define $E(a)$ as follows:

\be

  \item If $V_\phi(a)$ is a construction that
    represents a type and $H^{-1}(V_\phi(a))$ is eval-free, then 
    $E(a) = (\mname{eval},a,\mname{type}).$

  \item If $V_\phi(a)$ is a construction that represents a term,
    $H^{-1}(V_\phi(a))$ is eval-free, and
    $V_{\phi}(H^{-1}(V_{\phi}(a))) \not= \Undefined$, then $E(a) =
    (\mname{eval},a,\mname{C}).$

  \item If $V_\phi(a)$ is a construction that
    represents a formula and $H^{-1}(V_\phi(a))$ is eval-free, then 
    $E(a) = (\mname{eval},a,\mname{formula}).$

  \item Otherwise, $E(a)$ is undefined.

\ee

\begin{thm}
\bsp $F = (D_{\rm syn}, H, \sT_{\rm syn}, Q, E)$ is a syntax framework
for $(\sE_L,I)$.\esp
\end{thm}

\begin{proof}
$F$ is a syntax framework since it satisfies the four conditions of
  the Definition~\ref{df:syn-frame-lang}:

\be

  \item $H$ maps each $e \in \sE_L$ to a construction that represents
    the syntactic structure of $e$.  Thus $D_{\rm syn}$ is a set of
    values that represent syntactic structures and $H: \sE_L \tarrow
    D_{\rm syn}$ is injective and total.  So $R$ is a syntax
    representation of $\sE_L$.

  \item $I$ is an interpreted language.  $\sE_L \subseteq \sE_L$.
    $\sT_{\rm syn} \subseteq \sE_L$.  $D_{\rm syn} \subseteq D_M$
    (since since $D_{\rm syn}$ is the range of $H$, $H \mcolon \sE_L
    \tarrow \Dv$, and $\Dv \subseteq D_M$).  And $V_\phi$ restricted
    to $\sT_{\rm syn}$ is a total function $V' : \sT_{\rm syn} \tarrow
    D_{\rm syn}$.  So $(\sT_{\rm syn},I)$ is a syntax language for
    $R$.

  \item Let $e \in \sE_L$.  Then $V_\phi(Q(e)) =
    V_\phi((\mname{quote},e)) = H(e)$ by the definition of $Q$ and the
    definition of $V_\phi$ on quotations. So $Q : \sE_L \tarrow
    \sT_{\rm syn}$ is an injective, total function such that, for all
    $e \in \sE_L$, $V_\phi(Q(e)) = H(e)$.

  \item Let $a \in \sT_{\rm syn}$ such that $E(a)$ is defined.  Hence
    $V_\phi(a)$ is a construction that represents a type, term, or
    formula.  If $V_\phi(a)$ represents a type, term, or formula, let
    $k$ be \mname{type}, \mname{C}, or \mname{formula},
    respectively. Then $V_\phi(E(a)) = V_\phi((\mname{eval},a,k)) =
    V_{\phi}(H^{-1}(V_{\phi}(a)))$ by the definition of $E$ and the
    definition of $V_\phi$ on evaluations.  So $E : \sT_{\rm syn}
    \tarrow \sE_L$ is a partial function such that, for all $a \in
    \sT_{\rm syn}$, $V_\phi(E(a)) = V_\phi(H^{-1}(V_\phi(a)))$
    whenever $E(a)$ is defined.

\ee
Finally, $F$ is replete since $\sE_L$ is both the object and full
language of $F$ and $F$ has built-in quotation and evaluation.
\end{proof}

\bigskip

Quasiquotation is a notational definition in Chiron; it is not a
built-in operator in Chiron as quotation and evaluation
are~\cite{Farmer12}.  The quasiquotation defined in Chiron is
semantically equivalent to the notion of quasiquotation defined in the
previous section.

\newpage

\section{Conclusion} \label{sec:conclusion}

We have introduced a mathematical structure called a \emph{syntax
  framework} consisting of six major components:

\be

  \item A formal language $L$ with a semantics.

  \item A sublanguage $L_{\rm obj}$ of $L$ that is the object language
    of the syntax framework.

  \item A domain $D_{\rm syn}$ of values that represent the syntactic
    structures of expressions in $L_{\rm obj}$.

  \item A sublanguage $L_{\rm syn}$ of $L$ whose expressions denote
    values in $D_{\rm syn}$.

  \item A quotation function $Q : L_{\rm obj} \tarrow L_{\rm syn}$.

  \item An evaluation function $E : L_{\rm syn} \tarrow L_{\rm obj}$.

\ee

A syntax framework provides the means to reason about the interplay of
the syntax and semantics of the expressions in $L_{\rm obj}$ using
quotation and evaluation.  In particular, it provides three basic
syntax activities:

\be

  \item Expressing statements in $L$ about the syntax of $L_{\rm
    obj}$.

  \item Constructing expressions in $L_{\rm syn}$ that denote values
    in $D_{\rm syn}$.

  \item Employing expressions in $L_{\rm syn}$ as expressions in
    $L_{\rm obj}$.

\ee
These activities can be used to specify, and even implement,
transformers that map expressions in $L_{\rm obj}$ to expressions in
$L_{\rm obj}$.  They are needed, for example, to specify the rules of
differentiation and to prove that these rules correctly produce
representations of expressions that denote
derivatives~\cite{Farmer13}.  A syntax framework also provides a basis
for defining a notion of quasiquotation which is very useful for the
second basic activity.

When a syntax framework has built-in quotation and evaluation, it
provides the means to reason \emph{directly} in $L$ about the syntax
and semantics of the expressions in $L_{\rm obj}$.  However, in this
case, the evaluation function $E$ cannot be the direct evaluation
function (Lemma~\ref{lem:direct-eval}) and, if $L$ is sufficiently
expressive, $E$ cannot be total on quotations
(subsection~\ref{subsec:liar}) and thus the Law of Disquotation
(Lemma~\ref{lem:disquotation}) cannot hold universally. 

We showed that the notion of a syntax framework embodies a common
structure found in a variety of systems for reasoning about the
interplay of syntax and semantics.  We did this by showing how several
examples of such systems can be regarded as syntax frameworks.  Three
of these examples were the standard syntax-reasoning systems based on
expressions as strings, G\"odel numbers, and members of an inductive
type.  The other, more sophisticated, examples were taken from
the literature.

We have also mentioned that a syntax framework is not adequate for
modeling syntax reasoning in programming languages with mutable
variables.  This requires a generalization of a \emph{syntax
  framework} to a \emph{contextual framework} that will be presented
in a future paper.

\section*{Acknowledgments}
\bsp
The authors are grateful to Marc Bender, Jacques Carette, Michael
Kohlhase, Russell O'Connor, and Florian Rabe for their comments about
the paper.
\esp  

\iffalse

The authors would also like to thank the referees for their detailed
examination of the paper and valuable suggestions.

\fi

\bibliography{syntax-bibliography} 
\bibliographystyle{plain}

\end{document}

%% file: def.tex
\newcommand{\be}{\begin{enumerate}}
\newcommand{\ee}{\end{enumerate}}
\newcommand{\bi}{\begin{itemize}}
\newcommand{\ei}{\end{itemize}}
\newcommand{\bc}{\begin{center}}
\newcommand{\ec}{\end{center}}
\newcommand{\bsp}{\begin{sloppypar}}
\newcommand{\esp}{\end{sloppypar}}

\newcommand{\sglsp}{\ }
\newcommand{\dblsp}{\ \ }

\newcommand{\sA}{\mbox{$\cal A$}}

\newcommand{\sE}{\mbox{$\cal E$}}

\newcommand{\sM}{\mbox{$\cal M$}}

\newcommand{\sT}{\mbox{$\cal T$}}

\renewcommand{\phi}{\varphi}

\newcommand{\set}[1]{{\{ #1 \}}}

\newcommand{\mlist}[1]{{[ #1 ]}}

\newcommand{\bsynbrack}[1]{\lceil#1\rceil}

\newcommand{\mname}[1]{\mbox{\sf #1}}
\newcommand{\mcolon}{\mathrel:}
\newcommand{\mdot}{\mathrel.}

\newcommand{\tarrow}{\rightarrow}

\newcommand{\Neg}{\neg}

\newcommand{\ImpliesAlt}{\Rightarrow}

\newcommand{\ForallApp}{\forall\,}

\newcommand{\Undefined}{\bot}

\newcommand{\TRUE}{\mbox{{\sc t}}}
\newcommand{\FALSE}{\mbox{{\sc f}}}

\newcommand{\funapp}{\mathrel@}

\newcommand{\zf}{\mbox{\sc zf}}
\newcommand{\nbg}{\mbox{\sc nbg}}

\newtheorem{thm}{Theorem}[subsection]

\newtheorem{lem}[thm]{Lemma}

\newtheorem{rem}[thm]{Remark}
\newtheorem{eg}[thm]{Example}
\newtheorem{df}[thm]{Definition}

\newenvironment{proof}{\par\noindent{\bf Proof\dblsp}}{\hfill$\Box$}

\newcommand{\Dv}{D_{\rm v}}